\documentclass[a4paper,USenglish,numberwithinsect,cleveref, autoref, thm-restate]{lipics-v2021}

\pdfoutput=1 %
\hideLIPIcs  %
\nolinenumbers

\graphicspath{{./figures/}}%

\title{On the Complexity of Robust Markov Decision Processes and Bisimulation Metrics} %

\author{Marnix Suilen}{University of Antwerp, Belgium}{marnix.suilen@uantwerpen.be}{https://orcid.org/0000-0003-2163-3504}{}%

\author{Guillermo A. P\'{e}rez}{University of Antwerp -- Flanders Make, Belgium}{guillermo.perez@uantwerpen.be}{https://orcid.org/0000-0002-1200-4952}{}%

\authorrunning{M. Suilen and G.\,A. P\'{e}rez} %

\Copyright{Jane Open Access and Joan R. Public} %



\category{} %

\relatedversion{} %

\EventEditors{John Q. Open and Joan R. Access}
\EventNoEds{2}
\EventLongTitle{42nd Conference on Very Important Topics (CVIT 2016)}
\EventShortTitle{CVIT 2016}
\EventAcronym{CVIT}
\EventYear{2016}
\EventDate{December 24--27, 2016}
\EventLocation{Little Whinging, United Kingdom}
\EventLogo{}
\SeriesVolume{42}
\ArticleNo{23}

\usepackage{amsfonts}
\usepackage{mathtools}
\usepackage{amssymb}
\usepackage{amsthm}
\usepackage{xspace}
\usepackage[basic,caps]{complexity}
\usepackage{nicefrac}
\usepackage[inline]{enumitem}
\usepackage{booktabs}
\usepackage{algorithm}
\usepackage{algcompatible}

\usepackage{todonotes}

\usepackage{tikz}
\usetikzlibrary{patterns,arrows,arrows.meta,calc,shapes,shadows,decorations.pathmorphing,decorations.pathreplacing,automata,shapes.multipart,positioning,shapes.geometric,fit,circuits,trees,shapes.gates.logic.US,fit,backgrounds,decorations.text,topaths}

\usepackage{subcaption}

\newtheorem{problem}{Problem}

\newcommand{\ie}{\emph{i.e.}\@\xspace}
\newcommand{\eg}{\emph{e.g.}\@\xspace}

\newcommand{\tup}[1]{(#1)}
\newcommand{\supp}{\mathsf{Supp}}

\newcommand{\Linf}{L_\infty}
\newcommand{\Lone}{L_1}

\DeclareMathOperator*{\argmax}{argmax}

\renewcommand{\vec}[1]{\boldsymbol{#1}}

\newcommand{\bR}{\mathbb{R}}
\newcommand{\bE}{\mathbb{E}}

\newcommand{\cD}{\mathcal{D}}

\newcommand{\cU}{\mathcal{U}}
\newcommand{\cM}{\mathcal{M}}
\newcommand{\cR}{\mathcal{R}}

\newcommand{\cI}{\mathcal{I}}
\newcommand{\cLinf}{\mathcal{L}_\infty}

\newcommand{\cO}{\mathcal{O}}

\newcommand{\sinit}{s_{\iota}}

\newcommand{\Vpes}{\underline{V}}

\newcommand{\Deltapes}{\underline{\Delta}}

\newcommand{\Plow}{\check{P}}
\newcommand{\Pup}{\hat{P}}

\newcommand{\fM}{\mathfrak{M}}
\newcommand{\fD}{\mathfrak{D}}
\newcommand{\fF}{\mathfrak{F}}
\newcommand{\fB}{\mathfrak{B}}

\newcommand{\norm}[2]{\lVert #1 \rVert_{#2}}

\newcommand{\midd}{\,\middle\vert\,}

\begin{document}

\maketitle

\begin{abstract}
Robust Markov decision processes (RMDPs) extend standard Markov decision processes (MDPs) to account for uncertainty in the transition probabilities. 
RMDPs have an uncertainty set that defines a set of possible transition functions, each of which induces a standard MDP. 
The natural objective in an RMDP is to optimize the discounted cumulative reward under the worst-case transition function in the uncertainty set. 
We study the complexity of the associated threshold problem for RMDPs with polytopic uncertainty sets in halfspace representation. 
Previous results focused on approximating the optimum or restricted attention to specific subclasses of RMDPs, such as interval MDPs or $\Linf$-RMDPs. 
Our contributions are threefold: (1) For $(s,a)$-rectangular RMDPs, we prove that robust policy evaluation is in {\P} via robust linear programming, and that the threshold problem is in {\NP}. 
As a corollary, robust policy iteration is a polynomial-time algorithm for these RMDPs when the discount factor is fixed. 
(2) For 
$s$-rectangular RMDPs, we show that the threshold problem is in {\PSPACE}  via the first-order theory of the reals. 
(3) We establish lower bounds by reducing both parity games and bisimulation metrics between MDP states to the RMDP threshold problem.
A polynomial-time algorithm for the threshold problem would resolve the long-standing open question of whether parity games can be solved in polynomial time. 
The reduction from bisimulation metrics also yields a practical benefit: it allows us to apply robust policy iteration as a more efficient alternative to the standard fixed-point iteration, as our empirical evaluation demonstrates. 
\end{abstract}

\section{Introduction}\label{sec:introduction}

Markov decision processes (MDPs) form the theoretical backbone for many sequential decision-making problems and are foundational in reinforcement learning (RL)~\cite{DBLP:books/wi/Puterman94,DBLP:books/lib/SuttonB2018}.
The standard objective in an MDP is for the decision-making agent to compute a policy that maximizes the expected cumulative discounted reward.

MDPs rely on the assumption that their transition probabilities are exactly known; a prohibitive assumption in practice, where probabilities may be inaccurate or derived from a finite amount of samples and thus subject to statistical error~\cite{DBLP:journals/sttt/BadingsSSJ23}.
\emph{Robust MDPs} (RMDPs) overcome this assumption by instead modeling an uncertainty set which contains all probabilities the transition function can attain~\cite{DBLP:journals/mor/Iyengar05,DBLP:conf/birthday/SuilenBB0025,DBLP:journals/mor/WiesemannKR13}.
The typical objective in an RMDP is to optimize the policy and value against the \emph{worst-case} instance in the uncertainty set.
RMDPs are particularly well studied from the algorithmic and semantic perspectives, and have a wide range of applications in, \eg, reinforcement learning~\cite{DBLP:journals/jmlr/JakschOA10,DBLP:journals/tmlr/StarreLCO23,DBLP:journals/jcss/StrehlL08,DBLP:conf/nips/SuilenS0022}, statistical model checking~\cite{DBLP:conf/cav/AshokKW19}, and control~\cite{DBLP:journals/jair/BadingsRAPPSJ23}.
Under certain structural independence assumptions on the uncertainty set, known as \emph{rectangularity}, the Bellman operator of standard MDPs can be extended to RMDPs to account for the worst (or best) case instance.
Consequently, the classical value and policy iteration algorithms can be extended to RMDPs~\cite{DBLP:journals/mor/Iyengar05,DBLP:journals/ior/NilimG05}.

For standard MDPs, it is well known that comparing their value against a threshold is \P-hard and can be solved in polynomial time through a linear program (LP)~\cite{DBLP:journals/mor/PapadimitriouT87}.
Existing complexity results for RMDPs are less pronounced and follow different directions. 
In particular, previous works establish (1) polynomial time complexity based on the number of iterations the robust Bellman operator needs to achieve $\epsilon$-approximate solutions, where the precision $\epsilon$ and the discount factor are both considered constant~\cite{DBLP:journals/ior/NilimG05,DBLP:journals/mor/WiesemannKR13}; (2) that RMDPs with polytopic uncertainty where the vertices are part of the input (\ie, vertex representation) are in $\NP \cap \coNP$ by reduction to stochastic games~\cite{DBLP:conf/ijcai/ChatterjeeGK0Z24};
and (3) RMDPs with an $\Linf$-based uncertainty set can be solved in (strongly) polynomial time for a fixed discount factor.
Note that all these results are consistent with each other, but the complexity of exactly solving RMDPs with polytopic uncertainty in the halfspace representation, the input format typically used in practice~\cite{DBLP:journals/jair/BadingsRAPPSJ23,DBLP:conf/aaai/MeggendorferWW25,DBLP:journals/ior/NilimG05,DBLP:conf/nips/SuilenS0022}, has remained open until now.

\subparagraph*{Contributions}
In this paper, we expand on and generalize previous results regarding the computational complexity of RMDPs. We focus on the decision variant of the robust optimization problem: does there exist a policy such that its value surpasses a given threshold under the worst-case instance of the uncertainty set? 
Our specific contributions are as follows:
\begin{enumerate}
    \item For $(s,a)$-rectangular RMDPs, we prove that robust policy evaluation is in {\P} via robust linear programming (\Cref{thm:rmdp:policy:eval:in:P}), and that the threshold problem is in {\NP} (\Cref{thm:rmdp:sa:threshold:NP}). 
    As a corollary, robust policy iteration is a polynomial-time algorithm for these RMDPs when the discount factor is fixed (\Cref{thm:rpi:polytime}). 
    We also show how three common subclasses of RMDPs, interval MDPs, $\Lone$-RMDPs, and $\Linf$-RMDPs can be translated into equivalent polytopic RMDPs of polynomial size.
    \item For the larger class of $s$-rectangular RMDPs, where optimal policies may need randomization, we show that the threshold problem is in {\PSPACE} via the first-order theory of the reals (\Cref{thm:rmdp:s:threshold}).
    \item We establish lower bounds by reducing both parity games and bisimulation metrics between MDP states to the threshold problem (\Cref{prop:lower-bound,thm:bisim:reduces:rmdp}).
    Consequently, a polynomial-time algorithm for the RMDP threshold problem would resolve the long-standing open question of whether parity games can be solved in polynomial time. 
\end{enumerate}
We empirically demonstrate the benefits of using robust policy iteration over the standard fixed-point iteration for computing bisimulation metrics.

\subsection*{Related Work}
In addition to the results already mentioned, our work is consistent with and expands on the following existing results.
In~\cite{DBLP:conf/cav/PuggelliLSS13}, robust linear programs are constructed to compute \emph{optimistic} values and policies for RMDPs with convex uncertainty sets. 
That is, policies and their value under the \emph{best-case} in the uncertainty set.

For Markov chains with interval uncertainty,~\cite{chk2013} shows that computing the value can be done in polynomial time via a linear program.
Similarly, bisimulation metrics between Markov chains are also known to be in \P~\cite{DBLP:conf/fossacs/ChenBW12}.
Both results are special cases of ours: that evaluation of a memoryless policy in $(s,a)$-rectangular RMDPs is in \P.
For bisimulation metrics between MDP states, it is known that a single step of the fixed-point iteration is in \P~\cite{DBLP:conf/fsttcs/ChatterjeeAMR08}.
In~\cite{DBLP:conf/uai/FernsP14}, it is shown that, for a fixed coupling, the fixed-point operator for the bisimulation metric under that coupling yields an MDP value function. 
Consequently, there exists an MDP whose value function equals the optimal bisimulation metric.
We extend their result by showing that the optimal bisimulation metric, characterized by its fixed-point equation, coincides with the robust value function of a specifically constructed RMDP, characterized by its robust Bellman equation. This generalizes their result and allows us to directly apply RMDP algorithms such as robust policy iteration.

RMDPs are closely related to stochastic games~\cite{DBLP:conf/ijcai/ChatterjeeGK0Z24,DBLP:journals/mor/Iyengar05,DBLP:conf/birthday/SuilenBB0025}.
A similar observation has been made for MDPs (\ie, one-and-a-half-player games) and bisimulation metrics between DTMC states.
Consequently, MDP algorithms such as policy iteration have been proposed to compute said bisimulation metrics~\cite{DBLP:conf/concur/TangB16}. Our work connecting bisimulation metrics between MDP states with RMDPs can be seen as an extension of that work to two-and-a-half-player games.
Finally, in~\cite{DBLP:conf/concur/Kiefer024}, the problem of \emph{minimizing} the bisimulation metric under memoryless randomized policies is shown to be complete for the existential theory of the reals.

\section{Background}\label{sec:background}
We write $\bR$ and $\bR_{\geq 0}$ for the sets of real numbers and non-negative real numbers, respectively.
For a finite set $X$ we write $|X|$ for the number of elements in $X$.
A discrete probability distribution over $X$ is a function $\mu \colon X \to [0,1] \subset \bR$ with $\sum_{x \in X} \mu(x) = 1$.  
The set of all distributions over $X$ is denoted by $\cD(X)$.

Vectors and matrices are denoted in bold, \eg, $\vec{x} \in \bR^{n}$, and we write $\vec{x}(i)$ for the $i$-th element of $\vec{x}$.
The vector $\vec{e}_i \in \bR^{n}$ is the element of the standard basis with $\vec{e}_i(i) = 1$ and zeros elsewhere, and $\vec{I} \in \bR^{n \times n}$ denotes the identity matrix.
For a finite set $J$, we write $\vec{x} \in \bR^{J}$ for the vector indexed by elements of $J$, implicitly assuming a total order on $J$. 
This notation extends to Cartesian products of indices, \ie, $\vec{x} \in \bR^{J
\times K}$.
A \emph{convex polytope} 
is a closed and bounded set of points $\vec{x} \in \bR^{n}$ defined by a system of linear inequalities $\vec{D}\vec{x} \leq \vec{b}$, with $\vec{D} \in \bR^{m \times n}$ and $\vec{b} \in \bR^{m}$, which we denote by the tuple $(\vec{D},\vec{b})$.
For any bounded set $X$, its \emph{convex hull} is the smallest convex polytope that contains $X$.

A \emph{linear program} (LP) is an optimization problem of the form
\[
\begin{array}{ll}
    \text{minimize} & \vec{c}^\intercal \vec{x} \\
    \text{subject to} & \vec{A}\vec{x} \leq \vec{b}
\end{array}
\]
where $\vec{A} \in \bR^{m \times n}$, $\vec{b} \in \bR^m$, and $\vec{c} \in \bR^{n}$ form the input, and $\vec{x} \in \bR^n$ is a vector of $n$ optimization variables. \emph{Solving the LP} means finding an assignment of $\vec{x}$ that is feasible, \ie, it satisfies $\vec{A}\vec{x} \leq \vec{b}$, and optimal, \ie $\vec{c}^\intercal \vec{x} \leq \vec{c}^\intercal \vec{y}$ for all $\vec{y} \in \bR^n$ that are feasible.

For computability matters, we suppose all entries of $\vec{A}$, $\vec{b}$, and $\vec{c}$ are rational numbers. In this case, it is known that a solution to the LP, if one exists, is also a vector of rationals. LPs can be solved in 
time $(m+n+L)^{\cO(1)}$ \cite{DBLP:journals/ior/BlandGT81,khachiyan1980polynomial,DBLP:conf/stoc/Vaidya87}. In words, we can solve them in time polynomial in the number $n$ of variables, the number $m$ of inequalities, and the number $L$ of bits used to represent any entry of $\vec{A}$, $\vec{b}$, and $\vec{c}$ as a pair of integers written in binary. 

\subsection{Markov Decision Processes}

\begin{definition}[MDP]
    A Markov decision process (MDP) is a tuple $\tup{S,A,P,R,\sinit,\gamma}$, where $S$ and $A$ are finite sets of states and actions; $P \colon S \times A \to \cD(S)$, the transition function; $R \colon S \times A \to \bR$, the reward function; $\sinit \in S$, the initial state; $\gamma \in (0,1)$, a discount factor.
\end{definition}

A \emph{policy} selects actions to resolve the non-determinism in an MDP. 
In general, a policy maps paths to distributions over actions: $\pi \colon (S A)^* S \to \cD(A)$. The set of all policies is denoted by $\Pi$.
A \emph{memoryless randomized} policy is a function $\pi \colon S \to \cD(A)$ that maps states to distributions over actions; a \emph{memoryless deterministic} policy, a function $\pi \colon S \to A$ that maps states to actions.
The set of all memoryless randomized (resp. memoryless deterministic) policies is denoted by $\Pi^\text{MR}$ (resp. $\Pi^\text{MD}$).

An MDP and a policy together induce a classical (possibly infinite) Markov chain. 
The probability space of a Markov chain is well (and even uniquely) defined: its events are sets of infinite paths starting from the initial state $\sinit$ \cite{DBLP:books/daglib/0020348,klenke2008probability}. 
We can thus talk about the expectation of (Borel-measurable) aggregation functions applied to the transition rewards along runs.

\subparagraph{Discounted-reward value.} 
Given an MDP $M$ and a policy $\pi \in \Pi$, its expected discounted cumulative reward value function $V^\pi_M \colon S \to \bR$, or simply ``value'', is defined as
\[
V^\pi_M = \bE_{\pi} \left[ \sum_{t=0}^{\infty} \gamma^t R(s_t,a_t) \right],
\]
where $R(s_t,a_t)$ is the reward obtained from the state-action pair visited at step $t$.
We are interested in maximizing this value, %
\ie we are looking for the optimal value
\(
V^*_M = \sup_{\pi \in \Pi} V^{\pi}_M. %
\)

It is well established that for this objective in MDPs, the optimum can be achieved by a memoryless deterministic policy~\cite{DBLP:books/wi/Puterman94}.
Consequently, the supremum is equal to the maximum over the finite set of all memoryless deterministic policies, \ie 
\[
V^*_M = \max_{\pi \in \Pi^\text{MD}} V^\pi_M.
\]

\subsubsection{Decision Problem, Encoding, and Known Complexity Results}

We reformulate the problem of optimizing the expected reward as that of deciding whether a specific threshold can be surpassed or not.
\begin{problem}[MDP discounted reward problem]\label{problem:mdp:threshold}
        Given an MDP $M$, is there a policy whose value is at least some given threshold $\kappa \in \bR$, \ie, 
        $
        \exists \pi \in \Pi^\text{MD} \colon V_{M}^{\pi}(s_\iota) \geq \kappa
        $?
\end{problem}

It is well known that the discounted reward problem for MDPs reduces, in polynomial time, to solving an LP. To facilitate the encoding, the transition and reward functions of the MDP can be interpreted as matrices.
Formally, for each state-action pair $(s,a) \in S \times A$, we define probability vectors
$\vec{p}_{sa} \in [0,1]^{S}$ with $\vec{p}_{sa}(s') = P(s,a)(s')$ for all states $s' \in S$. Similarly, for any memoryless deterministic policy $\pi \in \Pi^\text{MD}$, we write $\vec{P}^\pi \in [0,1]^{S \times S}$ and $\vec{R}^\pi \in \bR^S$ to denote the transition matrix and reward vector of the induced Markov chain.

\begin{definition}[MDP LP encoding~\cite{DBLP:books/wi/Puterman94}]\label{def:mdp-lp}
Let $\vec{v} \in \bR^{S}$ be a vector of $|S|$ optimization variables, $\vec{c} \in \bR^{S}$ a vector of $|S|$ strictly positive, but otherwise arbitrary, coefficients (\eg, all ones). 
For
the (unique) 
optimal solution $\vec{v}$ to the following LP, we have that $\vec{v}(\iota) = V_M^*$.
\[
    \begin{array}{ll}
        \text{minimize} & \vec{c}^\intercal \vec{v} \\
        \text{subject to} & 
        \bigwedge_{s \in S} \bigwedge_{a \in A} (-\vec{e}_s + \gamma \vec{p}_{sa})^\intercal \vec{v} \leq - R(s,a)
    \end{array}
    \]
\end{definition}

This LP is polynomial in the size of $M$. Moreover, if the probabilities and rewards of $M$, and the threshold $\kappa$ are all rational numbers given as pairs of binary-encoded integers, the same holds for the entries of the matrices and vectors of the LP. (Importantly, the magnitude of the maximal constant in the LP is bounded by that of $M$.)
Since solving LPs is in polynomial time and $V^*_M = \max_{\pi \in \Pi^\text{MD}} V^\pi_M$, solving \autoref{problem:mdp:threshold} is also in polynomial time.

In practice, however, value iteration and policy iteration are more commonly used.
Value iteration, on the one hand, computes the least fixed point of the Bellman equation, which results in the optimal value function $V_M^*$.
Several variants of value iteration have been developed over time, either focusing on soundness~\cite{DBLP:journals/tcs/HaddadM18,DBLP:conf/cav/QuatmannK18} or efficiency~\cite{DBLP:conf/cav/HartmannsK20}.
Policy iteration, on the other hand, computes an optimal policy and its value by iteratively evaluating and improving an arbitrary initial policy. 
The evaluation of a memoryless deterministic policy $\pi$, \ie, computing $V^\pi_M$, can be done by solving the following system of linear equations:
\begin{equation}\label{eq:MDP:policy:eval}
    V^\pi_M = (\vec{I} - \gamma \vec{P}^\pi)^{-1} \vec{R}^\pi.
\end{equation}

The policy iteration algorithm (optimized with a so-called optimality test) is as follows.
\begin{definition}[Policy iteration with suboptimality test]\label{def:policy:iteration}
    Let $\pi \colon S \to A$ be an arbitrary policy and $A_s^* = A$ be the set of possibly optimal actions for every state $s$. %
    Repeat the following:
    \begin{enumerate}
        \item Evaluate the current policy $\pi$ by computing $V_{M}^\pi$ by solving \autoref{eq:MDP:policy:eval}. 
        \item Compute the normalized values
        \(
        \Delta^\pi_{sa} = \left(R(s,a) + \gamma \sum_{s' \in S} P(s,a)(s')V_M^\pi(s') \right) - V_M^{\pi}(s). 
        \)
        \item Determine the improving actions:
        $
        A_s^{+} = \{a \in A_s^* \mid \Delta^\pi_{sa} > 0\}
        $.
         If $A_s^+ = \emptyset$ for all $s \in S$, return $\pi$ and $V_M^\pi$. 
        Otherwise, update $\pi$ such that 
        $
        \forall s \in S \colon \pi(s) = \argmax_{a \in A_s^*} \Delta^{\pi}_{sa}$.
        
        \item %
        Prune suboptimal actions by computing the new set of possibly optimal actions as follows:
        \[
        A_s^* = \left\{a \in A_s^* \midd \Delta^\pi_{sa} \geq \max_{a' \in A} \{\Delta^\pi_{sa'}\} - \frac{\gamma}{1-\gamma} \left( \max_{s' \in S} \max_{a' \in A} \{\Delta^\pi_{s'a'}\} - \min_{s' \in S} \max_{a' \in A} \{\Delta^\pi_{s'a'}\} \right) \right\}.
        \]
    \end{enumerate}
\end{definition}

Policy iteration without suboptimality tests is sound as long as the policy evaluation step is precise enough to correctly identify improving actions~\cite{hartmanns2026revised}.
The variant with suboptimality tests is known to terminate in at most $T(|S|(|A|-1))$ iterations, with $T = \lfloor \nicefrac{|S|}{(1-\gamma)} \log (\nicefrac{|S|^2}{(1-\gamma)})  \rfloor + 1$, and $\cO(|S|^2(|S|(|A|-1)))$ operations per iteration to compute all values exactly.
Interestingly, it is unknown whether policy iteration runs in polynomial time independent of the discount factor, \ie, when it is part of the input~\cite{kallenberg2011markov}.

\section{Robust Markov Decision Processes}

We now introduce robust Markov decision processes (RMDPs).
Intuitively, an RMDP is a set of MDPs that differ in their transition probabilities, and policies for RMDPs must be optimized for the best- or worst-case MDP in the set.

\begin{definition}[RMDPs]
A \emph{robust MDP} (RMDP) is a tuple $\tup{S,A,\cU,R,\sinit,\gamma}$, where
$S,A,R,\sinit$ and $\gamma$ are as for MDPs, and $\cU \subseteq
\bR_{\geq 0}^{S \times A \times S}$ is a set of non-negative real vectors
$\vec{u} \in \cU$
that satisfy
$\sum_{s' \in S} \vec{u}(s,a,s') = 1$, for all $(s,a) \in S
\times A$.
\end{definition}
The \emph{uncertainty set} $\cU$ of an RMDP induces a family of transition functions $P_{\vec{u}}(s,a)(s') = \vec{u}(s,a,s')$, one for each $\vec{u} \in \cU$, such that each tuple $\tup{S,A,P_{\vec{u}},R,\sinit,\gamma}$ is a classical MDP.

\subsection{Rectangularity: Local Dependencies in Uncertainty Sets}

In this work, we focus on the case where $\cU$ is a convex polytope. We also study subcases where $\cU$ arises from the Cartesian product of ``transition-local'' polytopes. This additional structure is known as \emph{rectangularity}, with important algorithmic and complexity implications for solving the RMDP at hand.

\begin{figure}[t]
    \centering
    \begin{subfigure}[t]{0.3\textwidth}
        \centering
        \resizebox{\textwidth}{!}{
        \begin{tikzpicture}[>=stealth,
    action/.style={minimum size=1mm,inner sep=0pt,fill=black,circle}]
    \node[state] (s0) at (0,0) {$s_0$};
    \node[action] (s0a1) at ($(s0) + (0.8,0.6)$) {};
    \node[action] (s0a2) at ($(s0) + (0.8,-0.6)$) {};

    \node[state] (s1) at ($(s0) + (2.5,1)$) {$s_1$};
    
    \node[action] (s1a1) at ($(s1) + (1,0)$) {};
    \node[draw=none, right=-1mm of s1a1] (s1a1name){$a_1$};
    
    \node[state] (s2) at ($(s0) + (2.5,-1)$) {$s_2$};
    
    \node[state] (s3) at ($(s0) + (1.2,-2)$) {$s_3$};
    \node[action] (s3a2) at ($(s3) + (0,-1)$) {};
    \node[draw=none, below=-1mm of s3a2] (s3a2name) {$a_2$};

    \draw[<-] (s0.west) -- +(-0.3,0);
    \draw (s0) -- node[above]{$a_1$} (s0a1);
    \draw (s0) -- node[below]{$a_2$} (s0a2);

    \draw[->] (s0a1) edge[bend left=20]node[above]{$x_1$} (s1);
    \draw[->] (s0a1) edge[bend left=20]node[above]{$x_2$} (s2);

    \draw[->] (s0a2) edge[bend left=-20]node[below]{$x_3$} (s1);
    \draw[->] (s0a2) edge[bend left=-10]node[below]{$x_4$} (s2);
    \draw[->] (s0a2) edge[bend left=10]node[left]{$x_5$} (s3);

    \draw[->] (s2) edge[loop right]node[right]{$a_1, 1$} (s2);

    \draw (s1) -- (s1a1);
    \draw (s3) -- (s3a2);

    \draw[->] (s1a1) edge[bend right = 45]node[above]{$x_6$} (s1);
    \draw[->] (s1a1) edge[bend left = 20]node[right]{$x_7$} (s2);

    \draw[->] (s3a2) edge[bend left = 45]node[left]{$x_8$} (s0);
    \draw[->] (s3a2) edge[bend right = 45]node[right]{$x_9$} (s2);
    
\end{tikzpicture}
        }
        \caption{An $(s,a)$-rectangular RMDP. Each transition has a unique variable assigned to it. The constraints on the uncertainty set dictate the range of each variable so that they induce valid probability distributions for each state-action pair.}
        \label{fig:example:RMDP:sa}
    \end{subfigure}
    \hfil
    \begin{subfigure}[t]{0.3\textwidth}
        \centering
        \resizebox{\textwidth}{!}{
        \begin{tikzpicture}[>=stealth,
    action/.style={minimum size=1mm,inner sep=0pt,fill=black,circle}]
    \node[state] (s0) at (0,0) {$s_0$};
    \node[action] (s0a1) at ($(s0) + (0.8,0.6)$) {};
    \node[action] (s0a2) at ($(s0) + (0.8,-0.6)$) {};

    \node[state] (s1) at ($(s0) + (2.5,1)$) {$s_1$};
    
    \node[action] (s1a1) at ($(s1) + (1,0)$) {};
    \node[draw=none, right=-1mm of s1a1] (s1a1name){$a_1$};
    
    \node[state] (s2) at ($(s0) + (2.5,-1)$) {$s_2$};
    
    \node[state] (s3) at ($(s0) + (1.2,-2)$) {$s_3$};
    \node[action] (s3a2) at ($(s3) + (0,-1)$) {};
    \node[draw=none, below=-1mm of s3a2] (s3a2name) {$a_2$};

    \draw[<-] (s0.west) -- +(-0.3,0);
    \draw (s0) -- node[above]{$a_1$} (s0a1);
    \draw (s0) -- node[below]{$a_2$} (s0a2);

    \draw[->] (s0a1) edge[bend left=20]node[above]{${y}$} (s1);
    \draw[->] (s0a1) edge[bend left=20]node[above]{$x_2$} (s2);

    \draw[->] (s0a2) edge[bend left=-20]node[below]{${y}$} (s1);
    \draw[->] (s0a2) edge[bend left=-10]node[below]{$x_4$} (s2);
    \draw[->] (s0a2) edge[bend left=10]node[left]{$x_5$} (s3);

    \draw[->] (s2) edge[loop right]node[right]{$a_1, 1$} (s2);

    \draw (s1) -- (s1a1);
    \draw (s3) -- (s3a2);

    \draw[->] (s1a1) edge[bend right = 45]node[above]{$x_6$} (s1);
    \draw[->] (s1a1) edge[bend left = 20]node[right]{$x_7$} (s2);

    \draw[->] (s3a2) edge[bend left = 45]node[left]{$x_8$} (s0);
    \draw[->] (s3a2) edge[bend right = 45]node[right]{$x_9$} (s2);
    
\end{tikzpicture}
        }
        \caption{An $s$-rectangular RMDP. The same variable $y$ occurs on two different transitions. While there is a dependency between these two transitions, $y$ occurs only on transitions from the same state $s_0$. This RMDP is thus $s$-rectangular.}
        \label{fig:example:RMDP:s}
    \end{subfigure}
    \hfil
    \begin{subfigure}[t]{0.3\textwidth}
        \centering
        \resizebox{\textwidth}{!}{
        \begin{tikzpicture}[>=stealth,
    action/.style={minimum size=1mm,inner sep=0pt,fill=black,circle}]
    \node[state] (s0) at (0,0) {$s_0$};
    \node[action] (s0a1) at ($(s0) + (0.8,0.6)$) {};
    \node[action] (s0a2) at ($(s0) + (0.8,-0.6)$) {};

    \node[state] (s1) at ($(s0) + (2.5,1)$) {$s_1$};
    
    \node[action] (s1a1) at ($(s1) + (1,0)$) {};
    \node[draw=none, right=-1mm of s1a1] (s1a1name){$a_1$};
    
    \node[state] (s2) at ($(s0) + (2.5,-1)$) {$s_2$};
    
    \node[state] (s3) at ($(s0) + (1.2,-2)$) {$s_3$};
    \node[action] (s3a2) at ($(s3) + (0,-1)$) {};
    \node[draw=none, below=-1mm of s3a2] (s3a2name) {$a_2$};

    \draw[<-] (s0.west) -- +(-0.3,0);
    \draw (s0) -- node[above]{$a_1$} (s0a1);
    \draw (s0) -- node[below]{$a_2$} (s0a2);

    \draw[->] (s0a1) edge[bend left=20]node[above]{${y}$} (s1);
    \draw[->] (s0a1) edge[bend left=20]node[above]{$x_2$} (s2);

    \draw[->] (s0a2) edge[bend left=-20]node[below]{$x_3$} (s1);
    \draw[->] (s0a2) edge[bend left=-10]node[below]{$x_4$} (s2);
    \draw[->] (s0a2) edge[bend left=10]node[left]{$x_5$} (s3);

    \draw[->] (s2) edge[loop right]node[right]{$a_1, 1$} (s2);

    \draw (s1) -- (s1a1);
    \draw (s3) -- (s3a2);

    \draw[->] (s1a1) edge[bend right = 45]node[above]{$x_6$} (s1);
    \draw[->] (s1a1) edge[bend left = 20]node[right]{$x_7$} (s2);

    \draw[->] (s3a2) edge[bend left = 45]node[left]{${y}$} (s0);
    \draw[->] (s3a2) edge[bend right = 45]node[right]{$x_9$} (s2);
    
\end{tikzpicture}
        }
        \caption{A non-rectangular RMDP. The variable $y$ occurs on multiple transitions that belong to different states, \ie, $s_0$ and $s_3$. Hence, the RMDP is neither $(s,a)$- nor $s$-rectangular}
        \label{fig:example:RMDP:non}
    \end{subfigure}
    \caption{An example RMDP under different forms of rectangularity.}
    \label{fig:example:RMDP}
\end{figure}
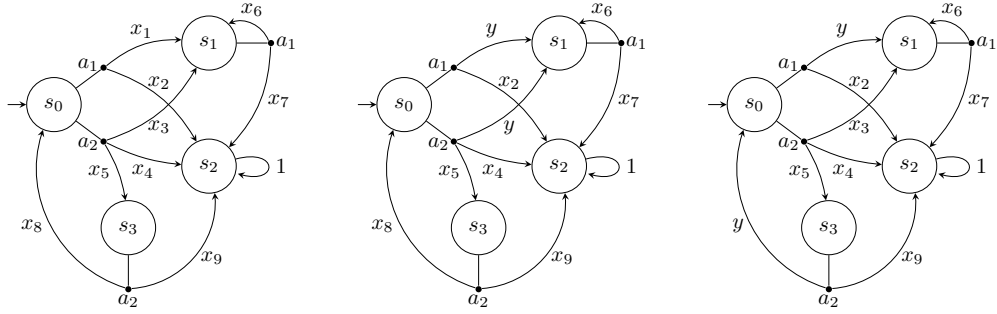

In a slight abuse of notation, for two sets $X \subseteq
\mathbb{R}_{\geq_0}^I$ and $Y \subseteq \mathbb{R}_{\geq 0}^J$ indexed by sets
$I$ and $J$ with $I \cap J = \emptyset$, we assume the vectors $\vec{v}$ in
the cross product $X \times Y$ are ``flattened'' so that $\vec{v} \in
\mathbb{R}_{\geq 0}^{I \cup J}$. That is, $X \times Y$ consists of real
vectors that are indexed by $I \cup J$ instead of pairs of vectors with the
first component indexed by $I$ and the other by $J$.

\begin{definition}[Rectangularity]
The uncertainty set $\cU$ of an RMDP $\tup{S,A,\cU,R,\sinit,\gamma}$ is:
\begin{description}
  \item[$(s,a)$-rectangular] if it is composed of sets $\cU_{(s,a)} \subseteq
    \mathbb{R}_{\geq 0}^{\{s\} \times \{a\} \times S}$, \ie\ $\cU =
    \bigtimes_{(s,a) \in S \times A} \cU_{(s,a)}$;
  \item[$s$-rectangular] if it is composed of sets $\cU_{s} \subseteq
    \mathbb{R}_{\geq 0}^{\{s\} \times A \times S}$ such that $\cU = \bigtimes_{s
    \in S} \cU_{s}$.
\end{description}
\end{definition}

Any RMDP that is neither $(s,a)$- nor $s$-rectangular is called \emph{non-rectangular}.
Intuitively, $(s,a)$-rectangularity allows nature to choose a probability distribution at each state-action pair independently from all other state-action pairs.
Consequently, nature can minimize the agent's reward by selecting probability distributions \emph{locally}, \ie, at each state-action pair independently and unrestricted from %
the uncertainty sets at other state-action pairs.
The $(s,a)$-rectangularity assumption is commonly made for both theoretical and algorithmic purposes, as well as in applications such as reinforcement learning~\cite{DBLP:journals/jmlr/JakschOA10,DBLP:conf/nips/SuilenS0022}, abstraction techniques~\cite{DBLP:journals/jair/BadingsRAPPSJ23}, and (statistical) model checking~\cite{DBLP:conf/cav/AshokKW19}.
Under $s$-rectangularity and, more generally, without rectangularity
assumptions, nature's choices are more restricted. While this results in
the agent's \emph{robust discounted reward} (a formal definition follows)
being less conservative compared to $(s,a)$-rectangularity, the associated
problems become harder.
\Cref{fig:example:RMDP} depicts an example RMDP under the three forms of rectangularity.

Known subclasses of $(s,a)$-rectangular RMDPs, such as interval MDPs (trivially included in convex polytopes), $\Linf$-RMDPs (which are a special case of IMDPs), and $\Lone$-RMDPs, can all be reduced to RMDPs with convex %
uncertainty sets, as detailed in \autoref{apx:standard:RMDPs}.

\subsection{Robust Discounted Reward}
Given an RMDP $\cM$ and a fixed policy $\pi \in \Pi$, its \emph{robust value} is defined as
\[
    \Vpes^{\pi}_{\cM} = \inf_{\vec{u} \in \cU} \bE_{\pi,P_{\vec{u}}} \left[ \sum_{t=0}^\infty \gamma^t R(s_t,a_t)  \right].
\]
In words, this is the worst-case value of $\pi$ over all instantiations of the uncertainty. As before, we are interested in maximizing this value, \ie\ we are looking for the optimal robust value $\Vpes^*_{\cM} = \sup_{\pi \in \Pi} \Vpes^{\pi}_{\cM}$.

\begin{proposition}[Sufficient Policy Classes~\cite{DBLP:journals/mor/WiesemannKR13}]\label{prop:policy:classes}
The following policy classes are known to be sufficient for optimality, depending on the rectangularity and uncertainty set assumed.
\begin{description}
    \item[$(s,a)$-rectangular.] Under $(s,a)$-rectangularity, memoryless deterministic policies suffice, \ie $\Vpes^*_{\cM} = \max_{\pi \in \Pi^\text{MD}} \Vpes^{\pi}_{\cM}$.
    \item[$s$-rectangular and convex.] Under $s$-rectangularity, if each composing uncertainty set $\cU_s$ is convex, memoryless randomized policies suffice, \ie $\Vpes^*_{\cM} = \max_{\pi \in \Pi^\text{MR}} \Vpes^{\pi}_{\cM}$.
\end{description}
\end{proposition}

To study its complexity, we reformulate the optimization problem as a decision one where we take a threshold as input.
\begin{problem}[Robust discounted reward problem]\label{problem:rmdp:threshold}
    Given an RMDP $\cM$, is there a policy whose robust value is at least some given threshold $\kappa \in \mathbb{R}$, \ie, 
    \(
    \exists \pi \in \Pi \colon \Vpes_{\cM}^{\pi}(\sinit) \geq \kappa?
    \)
\end{problem}
Note that we are not restricting $\pi$ in this general version of the problem. We will mostly look at particular instances based on rectangularity and properties of the uncertainty set, where further assumptions can be made about the policy.

\subsection{Bellman Operator for \texorpdfstring{$(s,a)$}{(s,a)}-Rectangular RMDPs}

For this section, we fix an RMDP $\cM$.
Under $(s,a)$-rectangularity, as per \autoref{prop:policy:classes}, the optimal robust value of $\cM$ is given as
\[
\Vpes^*_\cM = \max_{\pi \in \Pi^{\text{MD}}} \inf_{\vec{u} \in \cU} \bE_{\pi,P_{\vec{u}}}\left[ \sum_{t=0}^{\infty} \gamma^{t} R(s_t,a_t) \right].
\]

The robust value function $\Vpes^*_{\cM} \colon S \to \mathbb{R}$, where we are making the initial state $s_\iota$ a parameter, admits a \emph{robust version of the standard Bellman equation}. To state this formally, we recall
the robust Bellman operator $\fB \colon (S \to \bR) \to (S \to \bR)$:
\[
\fB(\Vpes)(s) = \max_{a \in A} R(s,a) + \gamma \inf_{\vec{u} \in \cU_{(s,a)}} \left\{ \sum_{s' \in S} P_{\vec{u}}(s,a)(s') \Vpes(s') \right\}.
\]
Intuitively, under $(s,a)$-rectangularity, nature's choice for a probability distribution can be inserted directly into the Bellman equation as a \emph{local optimization problem} at the given state $s$ and action $a$. The operator behaves just as nicely as that for standard MDPs.
\begin{proposition}[{\cite[Theorem 3.2]{DBLP:journals/mor/Iyengar05}}]\label{prop:robust-bellman}
    The operator $\fB$ is a contraction mapping w.r.t. $\Linf$ and its unique fixed point is $\Vpes^*_\mathcal{M}$.
\end{proposition}

\section{The Complexity of the Robust Discounted Reward Problem}
In this section, we consider the computational complexity of the robust
discounted reward problem for the different kinds of rectangularity
constraints we introduced earlier. First, we make some observations about how
hard the problem is: \P-hard and at least as hard as finding the winner in parity games, a
problem known to be in $\NP \cap \coNP$ and $\UP \cap \coUP$~\cite{DBLP:journals/ipl/Jurdzinski98} and even solvable in quasi-polynomial
time \cite{calude}, but not known to be in \P.
\begin{proposition}\label{prop:lower-bound}
    The robust discounted reward problem is \P-hard, even under rectangularity
    and convex uncertainty assumptions. Moreover, determining the winner of a
    parity game reduces, in polynomial time, to the robust discounted reward
    problem.
\end{proposition}
\begin{proof}[Proof sketch]
    The problem is \P-hard since the discounted reward problem for classical
    MDPs is \P-hard \cite{DBLP:journals/mor/PapadimitriouT87} and MDPs are a
    special case of (rectangular, convex) RMDPs.

    For the second part of the claim, consider an arbitrary \emph{discounted-sum
    game}---in contrast to MDPs, here, the successor state after playing $a$ from $s$ is resolved
    antagonistically by an adversary. We can construct an $(s,a)$-rectangular
    RMDP with uncertainty sets $\cU_{(s,a)}$ being the simplex of the deterministic distributions over successors of $s$ after playing $a$.
    By memoryless determinacy of discounted-sum games \cite{DBLP:journals/tcs/ZwickP96}, the robust value of the
    RMDP can be shown to coincide with the ``value'' of the discounted-sum
    game. Parity-game hardness is then obtained by appealing to the reduction from parity games to
    discounted-sum games~\cite{DBLP:journals/ipl/Jurdzinski98}.
    The full proof can be found in \Cref{apx:games}.
\end{proof}

\noindent
Concerning upper bounds, we establish that for $(s,a)$-rectangular RMDPs, the problem is in \NP. The crux of the argument relies on proving that the robust value of a (guessed) memoryless deterministic policy can be evaluated in polynomial time. For $s$-rectangular RMDPs, where optimal policies may need randomization, we instead resort to an encoding into the theory of the reals, placing it in \PSPACE.

To conclude, we leverage our policy evaluation results to derive a robust policy iteration procedure for $(s,a)$-rectangular convex-uncertainty RMDPs to obtain optimal robust policies.

\subsection{The Case of \texorpdfstring{$(s,a)$}{(s,a)}-Rectangularity and Convex Polytopic Uncertainty}
Let $\cM$ be an $(s,a)$-rectangular RMDP with convex uncertainty sets given as convex polytopes $\{ \cU_{(s,a)} = (\vec{D}_{sa},\vec{b}_{sa}) \subseteq \bR_{\geq 0 }^{\{s\} \times \{a\} \times S} \mid s \in S, a \in A\}$. For complexity matters, we suppose the entries of all $\vec{D}_{sa}$ and $\vec{b}_{sa}$ are rationals represented as pairs of binary-encoded integers. We also suppose we are given a threshold $\kappa \in \mathbb{Q}$. The last part of the input we consider is a
memoryless deterministic policy $\pi \colon S \to A$.

\begin{lemma}\label{thm:rmdp:policy:eval:in:P}
    Determining whether $\Vpes_{\cM}^{\pi}(\sinit) \geq \kappa$ holds true can be done in polynomial time.
\end{lemma}

We remark that this result generalizes \cite[Proposition 3]{chk2013}. In fact, it can be
proven by adapting the separation-oracle argument from \cite{chk2013}. Below,
we give an alternative argument using results from \emph{robust linear
optimization} \cite{GORISSEN2015124}.

\begin{proof}
    For all $s \in S$, let $\vec{u}_s$ be an arbitrary element of the uncertainty set $\cU_s = \bigtimes_{a \in A} \cU_{(s,a)}$. We introduce some notation to talk about the (uncertain) probabilities and rewards of the DTMC induced by $\pi$ and the $\vec{u}_s$.
    Write $\vec{p}^{\vec{u}\pi}_s$ for the vector from $\bR_{\geq 0}^S$ satisfying 
    $\vec{p}^{\vec{u}\pi}_s(s') = \vec{u}_s(s,\pi(s),s')$, for all $s' \in S$, and
    set $R^\pi(s) = R(s,\pi(s))$.

    We will be working with the definition of the robust value
    $\Vpes^\pi_{\cM} \colon S \to \bR$ where we consider it a function of the
    initial state. 
    Now, by definition of $\Vpes^\pi_{\cM}$ and linearity of
    the expectation operator, we have the following.
    \[
    \Vpes^{\pi}_{\cM}(s) = \inf_{\vec{u}_s \in \cU_s} R^\pi(s) + \gamma \sum_{s'
    \in S} \vec{p}^{\vec{u}\pi}_s(s') \Vpes^{\pi}_{\cM}(s'), \text{ for all } s \in S.
    \]
    If we write $P_{\vec{u}}^\pi$ for the matrix such that $P_{\vec{u}}^\pi(s)(s') = \vec{p}_s^{\vec{u}\pi}(s')$ for all $s,s' \in S$, then we can rephrase the above in matrix form as $\Vpes^\pi_{\cM} = \inf_{\vec{u} \in \bigtimes_{s \in S} \cU_s} R^\pi + \gamma P^\pi_{\vec{u}} \Vpes^\pi_{\cM}$. By definition of $\Vpes^\pi_{\cM}$ and $V^\pi_M$ for an MDP $M$, we also have that $\Vpes^\pi_{\cM} = \inf_{\vec{u} \in \bigtimes_{s \in S} \cU_s} V^\pi_{M(\vec{u})}$ where $M(\vec{u})$ is the MDP obtained from $\cM$ by fixing $\vec{u}$. Now, from \autoref{eq:MDP:policy:eval} we get that $V^\pi_{M(\vec{u})} = (\vec{I} - \gamma P^\pi_{\vec{u}})^{-1} R^\pi$. In particular, the matrix $\vec{I} - \gamma P^\pi_{\vec{u}}$ is guaranteed to be nonsingular, thus invertible, for all $\vec{u}$.
    
    Since each $\cU_s$ is a product of convex polytopes, hence also a convex polytope, we know that the infimum is achieved at some vertex. This means the infimum can be rewritten as a minimum over the (finite) set $\mathrm{vert}(\cU_s)$ of vertices of $\cU_s$.
    \[
    \Vpes^{\pi}_{\cM}(s) = \min_{\vec{u}_s \in \mathrm{vert}(\cU_s)} R^\pi(s) +
    \gamma \sum_{s' \in S} \vec{p}^{\vec{u}\pi}_s(s') \Vpes^{\pi}_{\cM}(s'), \text{ for
    all } s \in S.
    \]
    The equation above implies the following
    inequality. Moreover, we know that for all $s \in S$ there is some $\mu(s) \in
    \mathrm{vert}(\cU_s)$ that makes it an equality.
    \begin{equation}\label{eqn:feas-pol-eval}
      \Vpes^{\pi}_{\cM}(s) \leq R^\pi(s) +
      \gamma \sum_{s' \in S} \vec{p}^{\vec{u}\pi}_s(s') \Vpes^{\pi}_{\cM}(s'), \text{ for
      all } s \in S \text{ and all } \vec{u}_s \in \cU_s.
    \end{equation}
    
    Let $\vec{c} \in \bR^S_{> 0}$ and $\vec{v} \in \bR^S$ be a vector of
    (optimization) variables. The unique optimal solution $\vec{v}$ to the
    following mathematical program is such that $\vec{v}(s) =
    \Vpes^{\pi}_{\cM}(s)$, for all $s \in S$.
    \begin{equation}
    \begin{array}{ll}
        \text{maximize} & \vec{c}^\intercal \vec{v} \\
        \text{subject to} & \bigwedge_{s \in S} \bigwedge_{\vec{u}_s \in
        (\vec{D}_{s\pi(s)},\vec{b}_{s\pi(s)})}
        (\vec{e}_s - \gamma  \vec{u}_s(s,\pi(s)))^\intercal \vec{v} \leq R^\pi(s)
    \end{array} \label{eqn:robust-lp-eval}
    \end{equation}
    Indeed, \autoref{eqn:feas-pol-eval}, tells us $\vec{v}(s) =
    \Vpes^{\pi}_{\cM}(s)$ is a feasible assignment. Moreover, any optimal solution must satisfy $
        (\vec{I} - \gamma P_{\vec{u}}^\pi) \vec{v} \leq R^\pi$ for all $\vec{u} \in \bigtimes_{s \in S} \cU_s$. Then, $\vec{v} \leq  (\vec{I} - \gamma P_{\vec{u}}^\pi)^{-1}R^\pi$ for all $\vec{u} \in \bigtimes_{s \in S} \cU_s$, and in particular $\vec{v}(s) \leq ( (\vec{I} - \gamma P_{\mu(s)}^\pi)^{-1}R^\pi)(s) = \Vpes^\pi_{\cM}(s)$ for all $s \in S$, where the last equality follows from our choice of $\mu$ to make \autoref{eqn:feas-pol-eval} an equality. 
        Since $\vec{c} \in \bR_{> 0}^S$, we conclude $\Vpes^\pi_{\cM}$ is just as good as any optimal solution $\vec{v}$.

    \autoref{eqn:robust-lp-eval} is a  \emph{robust LP} with \emph{polytopic uncertainty}, already in standard form~\cite{GORISSEN2015124}.
    Hence, we can directly apply dualization techniques to replace the infinite set of constraints by a small number of constraints, yielding the following standard LP~\cite[Table 1]{GORISSEN2015124}.
    \begin{equation}
    \begin{array}{ll}
        \text{minimize} &-\vec{c}^\intercal \vec{v}\\
        \text{subject to} & \displaystyle\bigwedge_{ s \in S} \begin{cases}
            \vec{e}_s^\intercal \vec{v} + (\vec{b}_{s\pi(s)})^\intercal \vec{w}_s \leq R^\pi(s)  \\
            (-\vec{D}_{s\pi(s)})^\intercal \vec{w}_s = \gamma \vec{v} \\
            \vec{w}_s \geq 0 
        \end{cases} 
    \end{array}\label{eq:dtmc:robust:lp:classical-dual}
    \end{equation}
    This LP is polynomial in the input, particularly in the convex polytopes given in their halfspace representation as tuples $(\vec{D}_{s\pi(s)},\vec{b}_{s\pi(s)})$.
    Let $m$ be the maximum number of rows of any $\vec{D}_{s\pi(s)}$ and $\vec{b}_{s\pi(s)}$, for all $s$.
    The LP has at most $(m+1)|S|$ variables, $\vec{w}_s \in \bR^k$ with $k \leq m$ for each $s \in S$ and $\vec{v} \in \bR^S$, and $(2m+1)|S|$ constraints.
\end{proof}

Now, we can leverage the fact that deterministic policies suffice. This allows us to have a guess-and-check procedure for the robust discounted reward problem.

\begin{theorem}\label{thm:rmdp:sa:threshold:NP}
    The robust discounted reward problem (\autoref{problem:rmdp:threshold})  is in \NP{}, assuming $(s,a)$-rectangularity and convex polytopic uncertainty.
\end{theorem}

\begin{proof}
By \autoref{thm:rmdp:policy:eval:in:P} we have that robust policy evaluation of a memoryless deterministic policy is in \P.
We also know that memoryless deterministic policies are sufficient for discounted reward in $(s,a)$-rectangular RMDPs by \autoref{prop:policy:classes}.
Since each such policy can be represented using polynomial space, we can thus nondeterministically guess a policy and evaluate it in polynomial time.
We conclude that the problem is in \NP.
\end{proof}

\subsection{The Case of \texorpdfstring{$s$}{s}-Rectangularity and Convex Polytopic Uncertainty}
This time we let $\cM$ be an $s$-rectangular RMDP with convex uncertainty sets given as convex polytopes $\{\cU_s = (\vec{D}_s,\vec{b}_s) \subseteq \bR^{\{s\} \times A \times S} \mid s \in S\}$. Once more we suppose rationals given as pairs of binary-encoded integers, and we take a threshold $\kappa \in \mathbb{Q}$ as input.

\begin{theorem}\label{thm:rmdp:s:threshold}
    The robust discounted reward problem (\autoref{problem:rmdp:threshold}) is in \PSPACE{}, assuming $s$-rectangularity and convex polytopic uncertainty.
\end{theorem}
\begin{proof}
    By \autoref{prop:policy:classes}, it suffices to focus on memoryless randomized policies. Let $\pi \colon S \to \cD(A)$ be such a policy. Then, by repeating the first part of the proof of \autoref{thm:rmdp:policy:eval:in:P} for $s$-rectangular uncertainty sets, we establish that the unique optimal solution $\vec{v}$ to the following mathematical program is such that $\vec{v}(s) = \Vpes^\pi_{\cM}(s)$, for all $s \in S$.
    \begin{equation}
    \begin{array}{ll}
        \text{maximize} & \vec{c}^\intercal \vec{v} \\
        \text{subject to} & \bigwedge_{s \in S} \bigwedge_{\vec{u}_s \in
        (\vec{D}_{s},\vec{b}_{s})}
        (\vec{e}_s^\intercal - \gamma  \pi(s)^\intercal\vec{u}_s(s)) \vec{v} \leq R^\pi(s)
    \end{array} \label{eqn:robust-lp-srect-eval}
    \end{equation}
    \autoref{eqn:robust-lp-srect-eval} is a bit different from \autoref{eqn:robust-lp-eval} since $\pi$ is now memoryless randomized.
    Note that $\pi(s)$, interpreted as a vector $\bR^A$, is right-multiplied by $\vec{u}_s(s)$ interpreted as a matrix $\bR^{A \times S}$. Also, $R^\pi(s)$ is now set to $\sum_{a \in A} \pi(s)(a) R(s,a)$.

    It follows that the existence of a feasible assignment of $\vec{v}$ such that $\vec{v}(s_\iota) \geq \kappa$ suffices to conclude $\Vpes^\pi_{\cM}(s_\iota) \geq \kappa$. Recall that $\pi$ was chosen to be an arbitrary memoryless randomized policy. Towards encoding this in a logical statement, let $\vec{x} \in \bR^{S \times A}$ be a vector of $|S||A|$ real-valued variables. So far we have established that $\Vpes^*_{\cM}(s_\iota) \geq \kappa$ if and only if there exists an assignment of $\vec{v}$ and $\vec{x}$ that satisfies the following.
    \begin{gather*}
        \vec{v}(s_\iota) \geq \kappa \land \bigwedge_{s \in S} \left(\sum_{a \in A} \vec{x}(s)(a) = 1 \land \bigwedge_{a \in A} \vec{x}(s)(a) \geq 0\right) \text { and }\\
        \forall \vec{u} : \bigwedge_{s \in S} \left( \vec{D}_s \vec{u} \leq \vec{b}_s \rightarrow (\vec{e}_s^\intercal - \gamma \vec{x}(s)^\intercal \vec{u}(s))\vec{v} \leq \sum_{a \in A} R(s,a) \vec{x}(s)(a)\right).
    \end{gather*}
    This is a sentence in the first-order theory of the reals with order, multiplication, and addition. Given such a sentence of fixed quantifier alternation (in prenex form), determining its truth value is known to be in \PSPACE{}  \cite{basu2006algorithms,DBLP:journals/mst/SchaeferS24}. This concludes the proof of the claim, but we note that we even got a sentence from the $\exists \forall$-fragment of the theory, which is rather low in the recently explored hierarchy of complexity classes induced by the theory of the reals \cite{DBLP:journals/mst/SchaeferS24}.    
\end{proof}

\subsection{Robust Policy Iteration with Exact Evaluation}

Standard policy iteration can be extended to robust policy iteration for $(s,a)$-rectangular RMDPs with convex polytopic uncertainty by replacing the policy evaluation step from \autoref{def:policy:iteration} with the solution of the linear program from \autoref{eq:dtmc:robust:lp:classical-dual}.

\begin{definition}[Robust policy iteration]\label{def:robust:policy:iteration}
    Let $\pi \colon S \to A$ be an arbitrary policy and $A_s^* = A$ be the set of possibly optimal actions for every state $s$. %
    Repeat the following:
    \begin{enumerate}
        \item Evaluate the current policy $\pi$ by computing its robust value $\Vpes_{\cM}^\pi$ by solving \autoref{eq:dtmc:robust:lp:classical-dual}. 
        \item Compute the normalized values
        \begin{align*}
        \Deltapes^\pi_{sa} &= \left(R(s,a) + \gamma \inf_{\vec{u} \in \cU_{(s,a)}}\sum_{s' \in S} P_{\vec{u}}(s,a)(s')\Vpes_{\cM}^\pi(s')  \right) - \Vpes_{\cM}^{\pi}(s), %
        \end{align*}
        \item Determine the improving actions:
        $
        A_s^{+} = \{a \in A_s^* \mid \Deltapes^\pi_{sa} > 0\}
        $.
         If $A_s^+ = \emptyset$ for all $s \in S$, return $\pi$ and $\Vpes_{\cM}^\pi$. 
        Otherwise, update $\pi$ such that 
        $
        \forall s \in S \colon \pi(s) = \argmax_{a \in A_s^*} \Deltapes^{\pi}_{sa}$.
        
        \item %
        Prune suboptimal actions by computing the new set of possibly optimal actions as follows:
        \[
        A_s^* = \left\{a \in A_s^* \midd \Deltapes^\pi_{sa} \geq \max_{a' \in A} \left\{ \Deltapes^\pi_{sa'} \right\} - \frac{\gamma}{1-\gamma} \left( \max_{s' \in S} \max_{a' \in A} \left\{ \Deltapes^\pi_{s'a'} \right\} - \min_{s' \in S} \max_{a' \in A} \left\{ \Deltapes^\pi_{s'a'} \right\} \right) \right\}.
        \]
    \end{enumerate}
\end{definition}

\begin{theorem}\label{thm:rpi:polytime}
    Robust policy iteration returns an optimal policy $\pi \in \Pi^\text{MD}$ in at most $T(|S|(|A|-1))$ iterations, with $T = \lfloor \nicefrac{|S|}{(1-\gamma)} \log (\nicefrac{|S|^2}{(1-\gamma)}) \rfloor + 1$, and $(m+n+L)^{\cO(1)}$ operations per iteration, where $m$ is the number of inequalities, $n$ the number of variables, and $L$ the number of bits required to represent any of the rationals in the robust policy evaluation LP used in step one. 
\end{theorem}

\begin{proof}
    The number of iterations $T(|S|(|A|-1))$ follows directly from the number of iterations of policy iteration for standard MDPs, see \Cref{def:policy:iteration}.
    The complexity of a single iteration is dominated by solving $1+|S||A|$ polynomially-sized LPs, one global policy evaluation LP in step one, and a local LP for each state-action pair in step two.
    The number of variables and inequalities of the robust policy evaluation LP in step one is polynomial in the size of the sum of all local LPs, hence the runtime per iteration can be succinctly expressed as $(m+n+L)^{\cO(1)}$ steps in the input size of this global policy evaluation LP.
\end{proof}

Note that the algorithm runs in polynomial time if $\gamma$ is fixed, just like standard policy iteration for classical MDPs. Critically, this general formulation of robust policy iteration is \emph{not strongly polynomial}, even with $\gamma$ being fixed, because we solve an LP in the first steps (that is not known to be doable in strongly polynomial time).
If one can replace the LP from steps 1 and 2 by a strongly polynomial policy evaluation, like was recently done for RMDPs with bounded $\Linf$ uncertainty sets~\cite{asadi2026stronglypolynomialtimecomplexity}, then the whole algorithm becomes strongly polynomial---again, for a fixed $\gamma$.

The correctness of the optimality test follows from the robust Bellman operator (still) being a contraction mapping, and thus the proof for standard MDPs can be modified, see \Cref{apx:robust:optimality:test} for the full proof.

\section{The Relation with Bisimulation Metrics}
The preceding discussion on (robust) policy iteration is the perfect transition from RMDPs to \emph{bisimulation metrics} to compare MDPs. In this section, we will argue that computing the \emph{Kantorovich distance} between two given distributions with respect to a given \emph{pseudometric} roughly corresponds to the (internal) LPs solved in the first steps of the policy iteration procedure described above. Moreover, the fixed-point definition of bisimulation metrics we focus on mirrors the iterative nature of the robust policy iteration procedure. After some necessary preliminaries to introduce all the notation and formalize the two links just claimed, we will construct an RMDP from two given MDPs so that the robust value of the RMDP corresponds to the bisimulation metric between designated states of the MDPs.

\subsection{Notation for Bisimulation Metrics}

We recap the standard definitions for bisimulation metrics between (states of) MDPs~\cite{DBLP:conf/uai/FernsPP04}.

\begin{definition}[Pseudometric]
    A pseudometric on a set $X$ is a mapping $d \colon X \times X \to [0,\infty)$ such that for all $x,x',x'' \in X$:
    \begin{enumerate*}
        \item $x = x' \implies d(x,x') = 0$;
        \item $d(x,x') = d(x',x)$; and
        \item $d(x,x'') \leq d(x,x') + d(x',x'')$.
    \end{enumerate*}
    The set of all pseudometrics over $X$ is denoted as $\fM_X$.
\end{definition}

We compare two pseudometrics using the product ordering. In symbols, for any two pseudometrics $d_1, d_2 \colon X \times X \to [0,\infty)$, we write $d_1 \leq d_2$ if $d_1(x_1,x_2) \leq d_2(x_1,x_2)$ holds true for all $x_1,x_2 \in X$. More interestingly, we use pseudometrics to quantify the distance between distributions over $X$.

\begin{definition}[Kantorovich distance]\label{def:kantorovich:distance}
    Let $d \in \fM_X$ be a pseudometric over a finite set $X$, with $n = |X|$. The Kantorovich distance between two distributions $\mu,\nu \in \cD(X)$, written $\fD_K(d)(\mu,\nu)$, is the solution to the following LP:
    \[
    \begin{array}{ll}
        \text{minimize} & \sum_{j,k=1}^n \vec{\lambda}(j,k) d(x_j,x_k), \\
        \text{subject to} & 
        \left(\bigwedge_{j=1}^{n}  \sum_{k = 1}^n \vec{\lambda}(j,k) = \mu(x_j) \right) \land \left(
        \bigwedge_{k=1}^{n}  \sum_{j = 1}^n \vec{\lambda}(j,k) = \nu(x_k)\right) \land \vec{\lambda} \geq \vec{0},
    \end{array}
    \]
    where $\{\vec{\lambda}(j,k) \mid 1 \leq j,k \leq n\}$ are $n^2$ optimization variables.
    A feasible assignment ${\vec{\lambda}} \in \bR^{n \times n}$ to this LP is called a \emph{coupling}, and by minimizing the objective function, we find the \emph{minimal coupling} between distributions $\mu$ and $\nu$ w.r.t. the pseudometric $d$.
    We denote the set of couplings between $\mu$ and $\nu$ by $\Lambda_{\mu,\nu}$. %
\end{definition}
Note that the set of couplings $\Lambda_{\mu,\nu}$ is a convex polytope and $\sum_{j,k} \vec{\lambda}(j,k) = 1$, for all $\vec{\lambda} \in \Lambda_{\mu,\nu}$.

We can now define a (family of) bisimulation metrics between the states of an MDP
    $M = (S,A,P,R,\sinit,\gamma)$ using \cite[Theorem 4.5]{DBLP:conf/uai/FernsPP04}.

\begin{definition}[Fixed-point bisimulation metrics]\label{def:fixed:point:bisim:metric}
    Let $0 < c_R,c_T$ and $c_R + c_T \leq 1$ and consider the operator $\fF_K \colon \fM_S \to \fM_S$ on pseudometrics over the states $S$.
        \[
            \fF_K(d)(s_1,s_2) = \max_{a \in A} c_R |R(s_1,a) - R(s_2,a)| + c_T \fD_K(d)(P(s_1,a),P(s_2,a)).
        \]
    Then, $\fF_K$ has a least fixed point $d^*$ which we call the $(c_R,c_T)$-fixed-point bisimulation metric.
\end{definition}
    
    The bisimulation metric between two MDPs $M_1$ and $M_2$ is defined as the distance between states $d({s_{\iota_1}},{s_{\iota_2}})$ in the disjoint union of $M_1$ and $M_2$.
    A standard choice for the weights is $c_T = \gamma, c_R = 1-\gamma$~\cite{DBLP:conf/uai/FernsPP04}. Henceforth, we refer to the $(1-\gamma,\gamma)$-fixed-point bisimulation metric simply as the \emph{bisimulation metric}.

\subsection{From Bisimulation Metric to an \texorpdfstring{$(s,a)$}{(s,a)}-Rectangular RMDP}

We show how computing the bisimulation metric between states of an MDP reduces to computing the robust discounted reward value of an $(s,a)$-rectangular RMDP. 

\begin{definition}[Bisimulation metric RMDP]\label{def:bisim:RMDP}
    Let $M = \tup{S,A,P,R,\sinit,\gamma}$ be an MDP and consider
    states $s_1^*, s_2^* \in S$.
    Construct the following $(s,a)$-rectangular RMDP: $\cM = \tup{S \times S, A, \cU, \cR, \tup{s_1^*,s_2^*}, \gamma}$, where
    $\cU$ and $\cR$ are the uncertainty set and RMDP reward function defined as follows:
    \begin{itemize}
        \item $\cU = \bigtimes_{((s,s'),a) \in S \times S \times A} \cU_{((s,s'),a)} \text{ where }
        \cU_{((s,s'),a)} = \Lambda_{P(s,a),P(s',a)}$; and
        \item $\cR(\tup{s,s'},a) = (1-\gamma)|R(s,a) - R(s',a)|$.
    \end{itemize}
    That is, $\cU$ is an $(s,a)$-rectangular uncertainty set and the local uncertainty set $\cU_{(s,s'),a}$ is the set of couplings between the distributions $P(s,a)$ and $P(s',a)$ in the MDP transition function.
    The reward function is the difference between rewards in the MDP, scaled by $1 - \gamma$. 
\end{definition}

The robust value function of an RMDP is, in general, not a pseudometric. 
For the specific bisimulation metric RMDP constructed above, we first show that its robust value function forms a pseudometric. Then, we prove that the robust value precisely coincides with the bisimulation metric of the original MDP.
\begin{lemma}\label{lemma:robust:bisim:value:is:pseudometric}
The robust value function $\Vpes^*_{\cM} \colon S \times S \to \bR$ of the RMDP $\cM$ constructed in \autoref{def:bisim:RMDP} is a pseudometric over the set of states $S$.
\end{lemma}
\begin{proof}[Proof sketch]
By induction, it follows that each property holds for every iteration of the robust Bellman operator of $\cM$. The triangle inequality follows from the fact that the Kantorovich distance is a pseudometric.
The full proof can be found in \autoref{apx:proof:robust:bisim:value:is:pseudometric}.
\end{proof}

\begin{theorem}[Correctness of RMDP construction]\label{thm:bisim:reduces:rmdp}
    The optimal robust value $\Vpes^*_{\cM}$ of the RMDP $\cM$ constructed in \autoref{def:bisim:RMDP} precisely coincides with the bisimulation metric $d^*$ between states of the MDP $M$, \ie, $\Vpes^*_{\cM} = d^*$ and in particular $\Vpes_{\cM}^*(s_1^*, s_2^*) = d^*(s_1^*,s_2^*)$.
\end{theorem}
\begin{proof}
    Let $M$ be an MDP for which we want to compute the bisimulation metric between two states $\tup{s_1^*, s_2^*}$, and let $\cM$ be the RMDP constructed in \autoref{def:bisim:RMDP}.
    By \autoref{prop:robust-bellman}, the robust Bellman operator $\fB$ has $\Vpes^*_{\cM}$ as unique fixed point.
    We show that $\fB$, when used on nonnegative functions over states of $\cM$, is equivalent to the operator $\fF_K$ on pseudometrics over states of the MDP $M$. It will thus follow that $\fF_K$ is a contraction mapping, it has a unique fixed point $d^*$, and that $d^*$ coincides with $\Vpes^*_{\cM}$ as required.

    To ease readability, we omit the brackets around pairs of states when used in functions.
    The robust Bellman operator for $\cM = \tup{S \times S,  A, \cU, \cR, \tup{s_1^*, s_2^*}, \gamma}$ is defined as follows, where $\Vpes$ is a function $\Vpes \colon S \times S \to \bR$.
    \begin{align*}
    &\fB(\Vpes)(s,s') = \max_{a \in A} \cR(s,s',a) 
     + \gamma \inf_{\vec{u} \in \cU_{(s,s',a)}} \left\{ \sum_{\tup{t,t'} \in S \times S}  P_{\vec{u}}(s,s',a)(t,t') \Vpes(t,t')  \right\}.
     \end{align*}
     We unfold the definition of $\cR$ and rewrite the uncertainty set with the set of couplings:
     \begin{align*}
    & = \max_{a \in A} (1-\gamma)|R(s,a)-R(s',a)| 
     + \gamma \inf_{\vec{u} \in \Lambda_{P(s,a),P(s',a)}} \left\{ \sum_{\tup{t,t'} \in S \times S}  P_{\vec{u}}(s,s',a)(t,t') \Vpes(t,t')  \right\}.
    \end{align*}
    Note that by definition of $\cU$, we have $P_{\vec{u}}(s,s',a)(t,t') = \vec{u}(s,s',a,t,t')$, and thus we obtain
    $P_{\vec{u}}(s,s',a)(t,t') = \vec{\lambda}(t,t')$:
    \begin{align*}
    & = \max_{a \in A} (1-\gamma)|R(s,a)-R(s',a)| + \gamma \inf_{\vec{\lambda} \in \Lambda_{P(s,a),P(s',a)}} \left\{ \sum_{\tup{t,t'} \in S \times S}  \vec{\lambda}(t,t') \Vpes(t,t')  \right\}.
    \end{align*}
    Since the inner optimization problem is a linear program on which the infimum is attained, and minimizers are guaranteed to exist, we may replace $\inf$ by $\min$ and obtain precisely the definition of the Kantorovich distance $\fD_K$ with respect to $\Vpes$:
        \begin{align*}
    & = \max_{a \in A} (1-\gamma)|R(s,a)-R(s',a)| + \gamma \min_{\vec{\lambda} \in \Lambda_{P(s,a),P(s',a)}} \left\{ \sum_{\tup{t,t'} \in S \times S}  \vec{\lambda}(t,t') \Vpes({t,t'})  \right\}\\
    & = \max_{a \in A} (1-\gamma)|R(s,a)-R(s',a)| + \gamma \fD_K(\Vpes)(P(s,a),P(s',a))\\
    & = \fF_K(\Vpes)(s,s').\qedhere
    \end{align*}
\end{proof}

\subsection{Experimental Evaluation}

Our theoretical results enable the use of robust policy iteration as an algorithm to compute the bisimulation metric between MDP states.
We empirically evaluate this method and compare it with the standard fixed point approach~\cite{DBLP:conf/uai/FernsPP04}, implemented through robust bounded value iteration (RBVI) for RMDPs~\cite{DBLP:conf/aaai/MeggendorferWW25}.
RBVI not only bounds the value from below, as in the standard fixed-point approach, but also from above, and terminates when the difference between these two bounds is less than a prespecified $\epsilon$.

Our implementation is written in Python and uses Gurobi as LP solver~\cite{gurobi}.\footnote{The full implementation will be made publicly available upon publication.} 
The experiments were conducted on an Intel Core Ultra 7 268V with 8 cores of $2.2\,$GHz base speed and 32\,GB RAM.
For robust policy iteration, we implement a version with the optimality test enabled (RPIOT) and one without (RPI).
We initialize both variants with a random initial policy, and thus repeat the RPI experiments for $10$ different seeds to account for this source of randomness, and report mean values $\pm$ standard deviation.
RBVI precision is set to $\epsilon = 1e^{-6}$.
We set a timeout of two hours.

\begin{table}[t]
    \centering
    \begin{tabular}{ll|rrlrl}
    \toprule
        \textbf{Map size} & \textbf{Algorithm} & \textbf{Value} & \textbf{Time (s)}& $\pm$SD & \textbf{\#\,Iterations}& [min,max]\\ 
        \midrule
        $5 \times 5$ & RBVI & $1.346629$ & $332.80$& & $143\phantom{.0}$&\\
        $5 \times 5$ & RPI & $1.346629$&  $15.20$&$ \pm 1.09$& $6.7$& $[6,7]$\\
        $5 \times 5$ & RPIOT & $1.346629$ & $12.14$& $\pm 0.69$ & $5.8$& $[5,6]$\\
         \midrule
        $10 \times 10$ & RBVI & $1.434019$ &  $5298.27$&& $143\phantom{.0}$ &\\
        $10 \times 10$ & RPI & $1.434019$ & $391.86$ & $\pm 22.30$& $7.8$& $[7,8]$\\
        $10\times 10$ & RPIOT & $1.434019$ & $313.83$ & $\pm 12.69$& $6.9$& $[6,7]$\\
        \midrule
        $20 \times 20$ & RBVI &  \texttt{Time Out} & -- && --&\\
        $20 \times 20$ & RPI &  \texttt{Mem Out} & -- & &--&\\
        $20 \times 20$ & RPIOT &  \texttt{Mem Out} & -- & &--&\\
         \bottomrule
    \end{tabular}
    \caption{Comparison of robust bounded value iteration with robust policy iteration with and without optimality test on three Frozen Lake instances.}
    \label{tab:exp:frozen:lake:bisim}
\end{table}
\begin{figure}[t]
    \centering
    \hfill
    \begin{subfigure}[t]{0.49\textwidth}
       \centering
        \includegraphics[width=\linewidth]{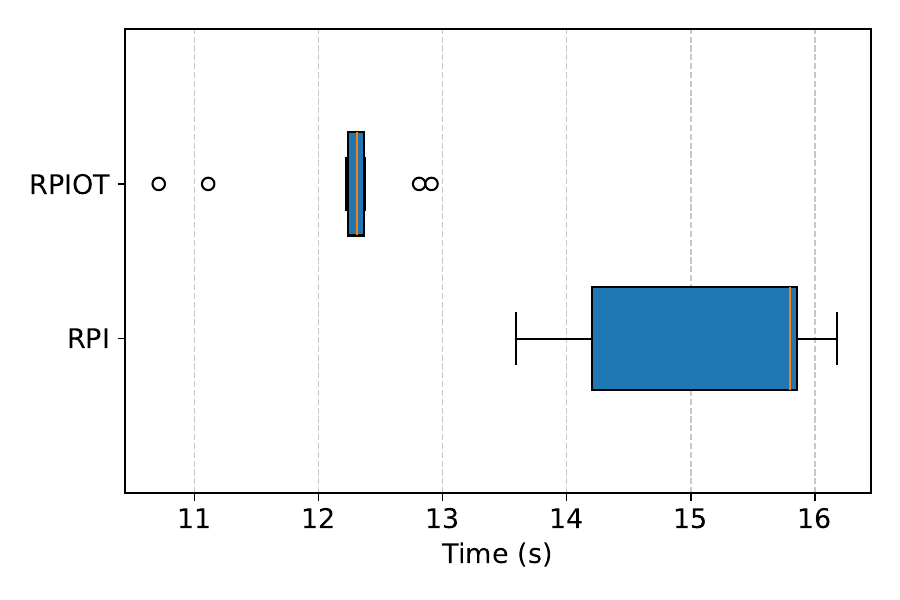}
                \caption{Frozen Lake $5 \times 5$.}
    \end{subfigure}
           \begin{subfigure}[t]{0.49\textwidth}
       \centering
        \includegraphics[width=\linewidth]{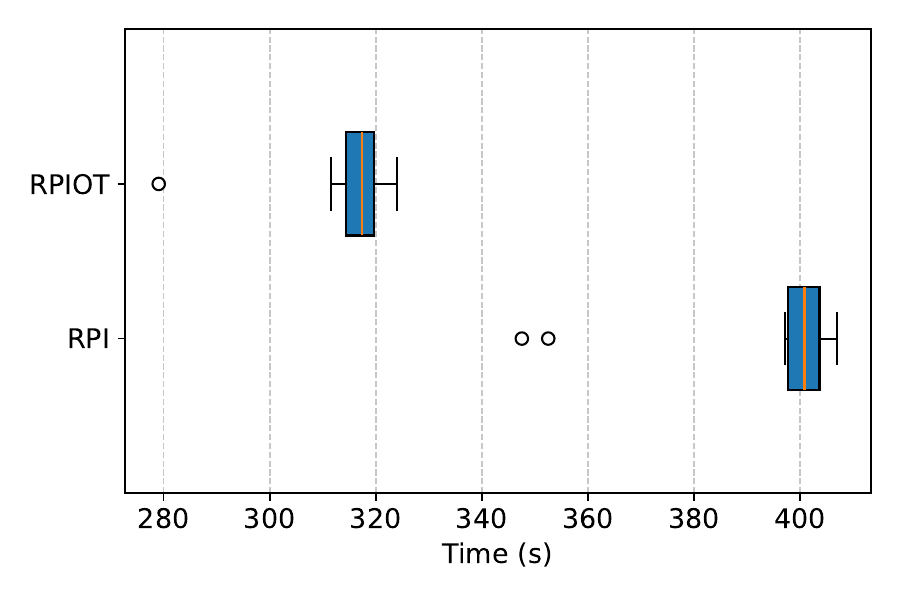}
        \caption{Frozen Lake $10 \times 10$.}
    \end{subfigure}
    \hfill
    \caption{Boxplots of the computation time of RPI and RPIOT across the $10$ different seeds.}
    \label{fig:boxplots}
\end{figure}

As MDP environments, we use three instances of Gymnasium Frozen Lake~\cite{DBLP:journals/corr/abs-2407-17032}, with map sizes of $5\times 5$, $10\times 10$, and $20 \times 20$, resulting in $25$, $100$, and $400$ state MDPs, respectively.
From the RMDP construction, it follows that to compute the bisimulation metric on an MDP with $|S|$ states, the RMDP has $|S|^2$ states; hence, our three Frozen Lake RMDP instances have $625$, $10\,000$, and $160\,000$ states, respectively.
The discount factor is set to $\gamma = 0.9$, and the rewards are $100$ for reaching the goal state, $-1$ for each step, and $-10$ for getting stuck in a hole.

We compute the bisimulation metric between states $0$ (the initial state) and $1$, the cell next to the initial state. Intuitively, the bisimulation metric between these two states indicates how much reward we can potentially gain by switching the initial state from $0$ to $1$.
The results are given in \Cref{tab:exp:frozen:lake:bisim}.
In \Cref{fig:boxplots}, we additionally compare the wall-clock computation time of RPI and RPIOT for all seeds. 
All algorithms achieve the same value to within 6 decimal places.
Robust policy iteration (RPI) performs significantly faster than robust bounded value iteration (RBVI), $22$ times faster on the small Frozen Lake map, and $13$ times faster on the $10 \times 10$ map.
Integrating the optimality test into RPI further increases its efficiency, as seen in the results for RPIOT.
For all seeds, the optimality test caused the algorithm to terminate one iteration earlier, which clearly affected the computation time.

On the $20 \times 20$ map, Gurobi runs out of memory during the first policy evaluation step of RPI, while RBVI reaches the timeout of two hours before convergence.
These last results signify the key difference between RBVI, which solves many small LPs per iteration, and RPI(OT), which solves one large LP per iteration.
Thus, there is a direct trade-off between computational efficiency and memory consumption.

\section{Conclusion}

We have clarified the complexity landscape of the broad class of RMDPs with polytopic uncertainty sets. 
In particular, we established that these $(s,a)$-rectangular RMDPs are $\P$-hard and in $\NP$ in general, and that robust policy iteration is a polynomial time algorithm for these RMDPs when the discount factor is fixed. 
Yet, the exact complexity of $(s,a)$-rectangular RMDPs remains open.
A potential direction for future work would be to prove hardness for some level in the hierarchy of the theory of the reals.
For $s$-rectangular RMDPs, we have established the complexity to be in $\PSPACE$ via the theory of the reals.
Reductions from parity games and bisimulation metrics relate RMDPs to other well-known problems and allow the use of robust policy iteration to compute bisimulation metrics between MDP states, providing a computationally efficient alternative to the standard fixed-point iteration.

In this paper, we considered only the expected discounted reward as the objective.
Our results on the complexity of policy evaluation naturally extend to other objectives, such as reachability, as long as the uncertainty set is \emph{graph-preserving}, \ie, all vectors in the uncertainty set must define the same graph structure, and memoryless policies suffice.
Whether our complexity results can be maintained without graph preservation, or whether it makes the RMDP threshold problem harder, remains open for future work.

\bibliography{references}

@book{klenke2008probability,
  title={Probability theory: a comprehensive course},
  author={Klenke, Achim},
  year={2008},
  publisher={Springer}
}

@book{basu2006algorithms,
  title={Algorithms in real algebraic geometry},
  author={Basu, Saugata and Pollack, Richard and Roy, Marie-Fran{\c{c}}oise},
  year={2006},
  publisher={Springer}
}

@article{DBLP:journals/mst/SchaeferS24,
  author       = {Marcus Schaefer and
                  Daniel Stefankovic},
  title        = {Beyond the Existential Theory of the Reals},
  journal      = {Theory Comput. Syst.},
  volume       = {68},
  number       = {2},
  pages        = {195--226},
  year         = {2024}
}

@article{calude,
  author       = {Cristian S. Calude and
                  Sanjay Jain and
                  Bakhadyr Khoussainov and
                  Wei Li and
                  Frank Stephan},
  title        = {Deciding Parity Games in Quasi-polynomial Time},
  journal      = {{SIAM} J. Comput.},
  volume       = {51},
  number       = {2},
  pages        = {17--152},
  year         = {2022}
}

@article{chk2013,
  author       = {Taolue Chen and
                  Tingting Han and
                  Marta Z. Kwiatkowska},
  title        = {On the complexity of model checking interval-valued discrete time
                  {M}arkov chains},
  journal      = {Inf. Process. Lett.},
  volume       = {113},
  number       = {7},
  pages        = {210--216},
  year         = {2013}
}

@book{DBLP:books/daglib/0020348,
  author       = {Christel Baier and
                  Joost{-}Pieter Katoen},
  title        = {Principles of model checking},
  publisher    = {{MIT} Press},
  year         = {2008}
}

@article{asadi2026stronglypolynomialtimecomplexity,
      title={Strongly Polynomial Time Complexity of Policy Iteration for ${L}_\infty$ Robust {MDP}s}, 
      author={Ali Asadi and Krishnendu Chatterjee and Ehsan Goharshady and Mehrdad Karrabi and Alipasha Montaseri and Carlo Pagano},
      year={2026},
      volume={abs/2601.23229},
      journal={CoRR}
}

@inproceedings{DBLP:conf/uai/FernsP14,
  author       = {Norman Ferns and
                  Doina Precup},
  title        = {Bisimulation Metrics are Optimal Value Functions},
  booktitle    = {{UAI}},
  pages        = {210--219},
  publisher    = {{AUAI} Press},
  year         = {2014}
}

@inproceedings{DBLP:conf/uai/FernsPP04,
  author       = {Norm Ferns and
                  Prakash Panangaden and
                  Doina Precup},
  title        = {Metrics for Finite {M}arkov Decision Processes},
  booktitle    = {{UAI}},
  pages        = {162--169},
  publisher    = {{AUAI} Press},
  year         = {2004}
}

@article{DBLP:journals/mor/Iyengar05,
  author       = {Garud N. Iyengar},
  title        = {Robust Dynamic Programming},
  journal      = {Math. Oper. Res.},
  volume       = {30},
  number       = {2},
  pages        = {257--280},
  year         = {2005}
}

@article{DBLP:journals/mor/WiesemannKR13,
  author       = {Wolfram Wiesemann and
                  Daniel Kuhn and
                  Ber{\c{c}} Rustem},
  title        = {Robust {M}arkov Decision Processes},
  journal      = {Math. Oper. Res.},
  volume       = {38},
  number       = {1},
  pages        = {153--183},
  year         = {2013}
}

@inproceedings{DBLP:conf/fossacs/ChenBW12,
  author       = {Di Chen and
                  Franck van Breugel and
                  James Worrell},
  title        = {On the Complexity of Computing Probabilistic Bisimilarity},
  booktitle    = {FoSSaCS},
  series       = {Lecture Notes in Computer Science},
  volume       = {7213},
  pages        = {437--451},
  publisher    = {Springer},
  year         = {2012}
}

@inproceedings{DBLP:conf/ijcai/ChatterjeeGK0Z24,
  author       = {Krishnendu Chatterjee and
                  Ehsan Kafshdar Goharshady and
                  Mehrdad Karrabi and
                  Petr Novotn{\'{y}} and
                  Dorde Zikelic},
  title        = {Solving Long-run Average Reward Robust {MDP}s via Stochastic Games},
  booktitle    = {{IJCAI}},
  pages        = {6707--6715},
  publisher    = {ijcai.org},
  year         = {2024}
}

@article{DBLP:journals/ior/NilimG05,
  author       = {Arnab Nilim and
                  Laurent El Ghaoui},
  title        = {Robust Control of {M}arkov Decision Processes with Uncertain Transition
                  Matrices},
  journal      = {Oper. Res.},
  volume       = {53},
  number       = {5},
  pages        = {780--798},
  year         = {2005}
}

@inproceedings{DBLP:conf/birthday/SuilenBB0025,
  author       = {Marnix Suilen and
                  Thom S. Badings and
                  Eline M. Bovy and
                  David Parker and
                  Nils Jansen},
  title        = {Robust {M}arkov Decision Processes: {A} Place Where {AI} and Formal
                  Methods Meet},
  booktitle    = {Principles of Verif. {(3)}},
  series       = {Lecture Notes in Computer Science},
  volume       = {15262},
  pages        = {126--154},
  publisher    = {Springer},
  year         = {2024}
}

@book{DBLP:books/wi/Puterman94,
  author       = {Martin L. Puterman},
  title        = {{M}arkov Decision Processes: Discrete Stochastic Dynamic Programming},
  series       = {Wiley Series in Probability and Statistics},
  publisher    = {Wiley},
  year         = {1994}
}

@misc{kallenberg2011markov,
  title={Lecture notes on {M}arkov decision processes},
  author={Kallenberg, Lodewijk},
  url={https://pub.math.leidenuniv.nl/~spieksmafm/colleges/LNMB-MDP/Lecture-notes-Kallenberg.pdf},
  year={2021}, 
  publisher={University of Leiden},
}

@inproceedings{DBLP:conf/stoc/Vaidya87,
  author       = {Pravin M. Vaidya},
  editor       = {Alfred V. Aho},
  title        = {An Algorithm for Linear Programming which Requires {$O(((m+n)n^2
                  + (m+n)^{1.5} n)L)$} Arithmetic Operations},
  booktitle    = {Proceedings of STOC,
                  1987, New York, New York, {USA}},
  pages        = {29--38},
  publisher    = {{ACM}},
  year         = {1987},
  url          = {https://doi.org/10.1145/28395.28399},
  doi          = {10.1145/28395.28399},
  bibsource    = {dblp computer science bibliography, https://dblp.org}
}

@article{DBLP:journals/ipl/Jurdzinski98,
  author       = {Marcin Jurdzinski},
  title        = {Deciding the Winner in Parity Games is in {UP} $\cap$
                  co-{UP}},
  journal      = {Inf. Process. Lett.},
  volume       = {68},
  number       = {3},
  pages        = {119--124},
  year         = {1998},
  url          = {https://doi.org/10.1016/S0020-0190(98)00150-1},
  doi          = {10.1016/S0020-0190(98)00150-1},
  bibsource    = {dblp computer science bibliography, https://dblp.org}
}

@article{DBLP:journals/tcs/ZwickP96,
  author       = {Uri Zwick and
                  Mike Paterson},
  title        = {The Complexity of Mean Payoff Games on Graphs},
  journal      = {Theor. Comput. Sci.},
  volume       = {158},
  number       = {1{\&}2},
  pages        = {343--359},
  year         = {1996},
  url          = {https://doi.org/10.1016/0304-3975(95)00188-3},
  doi          = {10.1016/0304-3975(95)00188-3},
  bibsource    = {dblp computer science bibliography, https://dblp.org}
}

@article{GORISSEN2015124,
title = {A practical guide to robust optimization},
journal = {Omega},
volume = {53},
pages = {124-137},
year = {2015},
issn = {0305-0483},
doi = {https://doi.org/10.1016/j.omega.2014.12.006},
url = {https://www.sciencedirect.com/science/article/pii/S0305048314001698},
author = {Bram L. Gorissen and İhsan Yanıkoğlu and Dick {den Hertog}},
}

@inproceedings{DBLP:conf/nips/SuilenS0022,
  author       = {Marnix Suilen and
                  Thiago D. Sim{\~{a}}o and
                  David Parker and
                  Nils Jansen},
  title        = {Robust Anytime Learning of {M}arkov Decision Processes},
  booktitle    = {NeurIPS},
  year         = {2022}
}

@article{DBLP:journals/jair/BadingsRAPPSJ23,
  author       = {Thom Badings and
                  Licio Romao and
                  Alessandro Abate and
                  David Parker and
                  Hasan A. Poonawala and
                  Mari{\"{e}}lle Stoelinga and
                  Nils Jansen},
  title        = {Robust Control for Dynamical Systems with Non-{G}aussian Noise via Formal
                  Abstractions},
  journal      = {J. Artif. Intell. Res.},
  volume       = {76},
  pages        = {341--391},
  year         = {2023}
}

@article{DBLP:journals/jmlr/JakschOA10,
  author       = {Thomas Jaksch and
                  Ronald Ortner and
                  Peter Auer},
  title        = {Near-optimal Regret Bounds for Reinforcement Learning},
  journal      = {J. Mach. Learn. Res.},
  volume       = {11},
  pages        = {1563--1600},
  year         = {2010}
}

@inproceedings{DBLP:conf/cav/AshokKW19,
  author       = {Pranav Ashok and
                  Jan Kret{\'{\i}}nsk{\'{y}} and
                  Maximilian Weininger},
  title        = {{PAC} Statistical Model Checking for {M}arkov Decision Processes and
                  Stochastic Games},
  booktitle    = {{CAV} {(1)}},
  series       = {Lecture Notes in Computer Science},
  volume       = {11561},
  pages        = {497--519},
  publisher    = {Springer},
  year         = {2019}
}

@article{DBLP:journals/mor/PapadimitriouT87,
  author       = {Christos H. Papadimitriou and
                  John N. Tsitsiklis},
  title        = {The Complexity of {M}arkov Decision Processes},
  journal      = {Math. Oper. Res.},
  volume       = {12},
  number       = {3},
  pages        = {441--450},
  year         = {1987}
}

@article{DBLP:journals/ior/BlandGT81,
  author       = {Robert G. Bland and
                  Donald Goldfarb and
                  Michael J. Todd},
  title        = {The Ellipsoid Method: {A} Survey},
  journal      = {Oper. Res.},
  volume       = {29},
  number       = {6},
  pages        = {1039--1091},
  year         = {1981}
}

@article{khachiyan1980polynomial,
  title={Polynomial algorithms in linear programming},
  author={Khachiyan, Leonid G.},
  journal={USSR Computational Mathematics and Mathematical Physics},
  volume={20},
  number={1},
  pages={53--72},
  year={1980},
  publisher={Elsevier}
}

@article{DBLP:journals/tcs/HaddadM18,
  author       = {Serge Haddad and
                  Benjamin Monmege},
  title        = {Interval iteration algorithm for {MDP}s and {IMDP}s},
  journal      = {Theor. Comput. Sci.},
  volume       = {735},
  pages        = {111--131},
  year         = {2018}
}

@inproceedings{DBLP:conf/cav/QuatmannK18,
  author       = {Tim Quatmann and
                  Joost{-}Pieter Katoen},
  title        = {Sound Value Iteration},
  booktitle    = {{CAV} {(1)}},
  series       = {Lecture Notes in Computer Science},
  volume       = {10981},
  pages        = {643--661},
  publisher    = {Springer},
  year         = {2018}
}

@inproceedings{DBLP:conf/cav/HartmannsK20,
  author       = {Arnd Hartmanns and
                  Benjamin Lucien Kaminski},
  title        = {Optimistic Value Iteration},
  booktitle    = {{CAV} {(2)}},
  series       = {Lecture Notes in Computer Science},
  volume       = {12225},
  pages        = {488--511},
  publisher    = {Springer},
  year         = {2020}
}

@article{hartmanns2026revised,
  title={The revised practitioner’s guide to {MDP} model checking algorithms},
  author={Hartmanns, Arnd and Junges, Sebastian and Quatmann, Tim and Weininger, Maximilian},
  journal={International Journal on Software Tools for Technology Transfer},
  pages={1--31},
  year={2026},
  publisher={Springer}
}

@inproceedings{DBLP:conf/cav/PuggelliLSS13,
  author       = {Alberto Puggelli and
                  Wenchao Li and
                  Alberto L. Sangiovanni{-}Vincentelli and
                  Sanjit A. Seshia},
  title        = {Polynomial-Time Verification of {PCTL} Properties of {MDP}s with Convex
                  Uncertainties},
  booktitle    = {{CAV}},
  series       = {Lecture Notes in Computer Science},
  pages        = {527--542},
  publisher    = {Springer},
  year         = {2013}
}

@book{DBLP:books/lib/SuttonB2018,
  author       = {Richard S. Sutton and
                  Andrew G. Barto},
  title        = {Reinforcement learning - an introduction, 2nd Edition},
  publisher    = {{MIT} Press},
  year         = {2018}
}

@article{DBLP:journals/sttt/BadingsSSJ23,
  author       = {Thom Badings and
                  Thiago D. Sim{\~{a}}o and
                  Marnix Suilen and
                  Nils Jansen},
  title        = {Decision-making under uncertainty: beyond probabilities},
  journal      = {Int. J. Softw. Tools Technol. Transf.},
  volume       = {25},
  number       = {3},
  pages        = {375--391},
  year         = {2023}
}

@inproceedings{DBLP:conf/aaai/MeggendorferWW25,
  author       = {Tobias Meggendorfer and
                  Maximilian Weininger and
                  Patrick Wienh{\"{o}}ft},
  title        = {Solving Robust {M}arkov Decision Processes: Generic, Reliable, Efficient},
  booktitle    = {{AAAI}},
  pages        = {26631--26641},
  publisher    = {{AAAI} Press},
  year         = {2025}
}

@inproceedings{DBLP:conf/concur/Kiefer024,
  author       = {Stefan Kiefer and
                  Qiyi Tang},
  title        = {Minimising the Probabilistic Bisimilarity Distance},
  booktitle    = {{CONCUR}},
  series       = {LIPIcs},
  pages        = {32:1--32:18},
  publisher    = {Schloss Dagstuhl - Leibniz-Zentrum f{\"{u}}r Informatik},
  year         = {2024}
}

@inproceedings{DBLP:conf/fsttcs/ChatterjeeAMR08,
  author       = {Krishnendu Chatterjee and
                  Luca de Alfaro and
                  Rupak Majumdar and
                  Vishwanath Raman},
  title        = {Algorithms for Game Metrics},
  booktitle    = {{FSTTCS}},
  series       = {LIPIcs},
  pages        = {107--118},
  publisher    = {Schloss Dagstuhl - Leibniz-Zentrum f{\"{u}}r Informatik},
  year         = {2008}
}

@inproceedings{DBLP:conf/concur/TangB16,
  author       = {Qiyi Tang and
                  Franck van Breugel},
  title        = {Computing Probabilistic Bisimilarity Distances via Policy Iteration},
  booktitle    = {{CONCUR}},
  series       = {LIPIcs},
  pages        = {22:1--22:15},
  publisher    = {Schloss Dagstuhl - Leibniz-Zentrum f{\"{u}}r Informatik},
  year         = {2016}
}

@article{DBLP:journals/corr/abs-2407-17032,
  author       = {Mark Towers and
                  Ariel Kwiatkowski and
                  Jordan K. Terry and
                  John U. Balis and
                  Gianluca De Cola and
                  Tristan Deleu and
                  Manuel Goul{\~{a}}o and
                  Andreas Kallinteris and
                  Markus Krimmel and
                  Arjun KG and
                  Rodrigo Perez{-}Vicente and
                  Andrea Pierr{\'{e}} and
                  Sander Schulhoff and
                  Jun Jet Tai and
                  Hannah Tan and
                  Omar G. Younis},
  title        = {Gymnasium: {A} Standard Interface for Reinforcement Learning Environments},
  journal      = {CoRR},
  volume       = {abs/2407.17032},
  year         = {2024}
}

@misc{gurobi,
  author = {{Gurobi Optimization, LLC}},
  title = {{Gurobi Optimizer Reference Manual}},
  year = 2026,
  url = "https://www.gurobi.com"
}

@article{DBLP:journals/jcss/StrehlL08,
  author       = {Alexander L. Strehl and
                  Michael L. Littman},
  title        = {An analysis of model-based Interval Estimation for {M}arkov Decision
                  Processes},
  journal      = {J. Comput. Syst. Sci.},
  volume       = {74},
  number       = {8},
  pages        = {1309--1331},
  year         = {2008}
}

@article{DBLP:journals/tmlr/StarreLCO23,
  author       = {Rolf A. N. Starre and
                  Marco Loog and
                  Elena Congeduti and
                  Frans A. Oliehoek},
  title        = {An Analysis of Model-Based Reinforcement Learning From Abstracted
                  Observations},
  journal      = {Trans. Mach. Learn. Res.},
  volume       = {2023},
  year         = {2023}
}

\clearpage
\newpage
\appendix

\section{Special RMDP Instances}\label{apx:standard:RMDPs}
In this appendix, we show how three commonly used instances of RMDPs, namely interval MDPs (IMDPs), and $(s,a)$-rectangular $\Lone$- and $\Linf$-RMDPs, are $(s,a)$-rectangular RMDPs with polytopic uncertainty sets.

\subsection{Interval MDPs}
We recap the standard definition of interval MDPs used across the literature~\cite{DBLP:conf/nips/SuilenS0022,DBLP:conf/birthday/SuilenBB0025} and explain how they are $(s,a)$-rectangular RMDPs with convex polytopic uncertainty sets.

\begin{definition}
    An interval MDP (IMDP) is a tuple $\cI = (S,A,\Plow,\Pup,R,\sinit,\gamma)$ where $S,A,R,\sinit$ and $\gamma$ are as for standard MDPs, and $\Plow, \Pup \colon S \times A \times S \to [0,1]$ are the lower and upper bound functions on the transition probabilities such that
    $0 \leq \Plow(s,a,s') \leq \Pup(s,a,s') \leq 1$.
\end{definition}
An IMDP is an $(s,a)$-rectangular RMDP with uncertainty set
\[
\cU = \bigtimes_{(s,a) \in S \times A} \cU_{(s,a)} = \left\{ \vec{u} \in \bR^{S}_{\geq 0} \midd \sum_{s' \in S} \vec{u}(s') = 1 \wedge \bigwedge_{s'\in S} \vec{u}(s') \in [\Plow(s,a,s'), \Pup(s,a,s')] \right\}.
\]

\begin{lemma}\label{lemma:IMDP:is:RMDP}
Every state-action uncertainty set $\cU_{(s,a)}$ of an IMDP forms a convex polytope.    
\end{lemma}

\begin{proof}
    Since all constraints in $\cU_{(s,a)}$ are linear, it follows that the set forms a convex polytope. 
    We now explicitly construct the matrix $\vec{D}
    _{sa}$ and vector $\vec{b}_{sa}$ such that $\cU_{s,a} = (\vec{D}_{sa},\vec{b}_{sa})$.
    Assume an ordering over the $k = |S|$ successor states and define
\[
\vec{D}_{sa} = 
\begin{pmatrix}
    {-}\vec{1}_k^\intercal \\
    {-}\vec{1}_k^\intercal \\
    \vec{1}_k^\intercal \\
    {-}\vec{I}_k \\
    \vec{I}_k
\end{pmatrix},
\qquad \vec{b}_{sa} = 
\begin{pmatrix}
    0 \\
    -1 \\
    1 \\
    -\Plow(s,a,s_1)\\
    \vdots\\
    -\Plow(s,a,s_k) \\
    \Pup(s,a,s_1) \\
    \vdots \\
    \Pup(s,a,s_k)
\end{pmatrix},
\]
where $\vec{1}_k$ is the one vector of size $k$, and $\vec{I}_k$ is the $k \times k$ identity matrix.
Clearly, constructing $\vec{D}_{sa}$ and $\vec{b}_{sa}$ from an input IMDP is polynomial.
\end{proof}

\subsection{\texorpdfstring{$\Linf$}{L-infinity}-RMDPs}

Next, we recap the standard definition for $\Linf$-RMDPs, and prove how every $\Linf$-RMDP can be transformed into an IMDP.
Recall that the $\Linf$ distance between two vectors $\vec{x}, \vec{y} \in \bR^{n}$ is defined as $\norm{\vec{x} - \vec{y}}{\infty} = \max_i \{|\vec{x}(i) - \vec{y}(i)| \}$.

\begin{definition}[$\Linf$-RMDP]
    An $\Linf$-RMDP is a tuple ${\L}_{\infty} = (M,\epsilon)$ where $M = (S,A,P,R,\sinit,\gamma)$ is an MDP, and $\epsilon \colon S \times A \to \bR_{\geq 0}$ is a state-action distance measure.
    The $\Linf$-RMDP $\cL_\infty$ is an $(s,a)$-rectangular RMDP with state-action uncertainty sets defined by $\cU_{(s,a)} = \{\vec{u} \in \bR^{S}_{\geq 0} \mid \sum_{s' \in S} \vec{u}(s') = 1 \wedge \norm{\vec{p}_{sa} - \vec{u}}{\infty} \leq \epsilon(s,a) \}$, \ie, the set of probability vectors within $\Linf$ distance $\epsilon(s,a)$ to the transition probability $P(s,a)$ of $M$.
\end{definition}

\begin{lemma}
$\Linf$-RMDPs reduce to IMDPs. That is, for every $\Linf$-RMDP, there is a polynomially sized IMDP with precisely the same uncertainty set.
\end{lemma}

\begin{proof}
    Given an $\Linf$-RMDP $\cLinf = (M,\epsilon)$ around MDP $M = (S,A,P,R,\sinit,\gamma)$, we construct the IMDP $\cI = (S,A,\Plow,\Pup,R,\sinit,\gamma)$.
    The interval on each transition $[\Plow(s,a,s'),\Pup(s,a,s')]$ is defined by
    \[
    \Plow(s,a,s') = \max\{P(s,a)(s') - \epsilon(s,a),0\}, \quad \Pup(s,a,s') = \min\{P(s,a)(s') + \epsilon(s,a), 1\}.
    \]

    Note that this is a valid IMDP satisfying $0 \leq \Plow(s,a,s') \leq \Pup(s,a,s') \leq 1$ for every transition, and the uncertainty sets at each state-action pair are
    \[
        \cU^{\cI}_{(s,a)} = \left\{ \vec{u} \mid \sum_{s' \in S} \vec{u}(s') = 1 \wedge \vec{u}(s') \in [\Plow(s,a,s'), \Pup(s,a,s')]  \right\}.
    \]
    By the definition of the intervals, we have
    \[
        \cU^{\cI}_{(s,a)} = \left\{ \vec{u} \midd \sum_{s' \in S} \vec{u}(s') = 1 \wedge \vec{u}(s') \in [P(s,a)(s') - \epsilon(s,a), P(s,a)(s') + \epsilon(s,a)] \wedge \vec{u}(s') \in [0,1]  \right\}.
    \]

    We now show $\cU^{\cLinf}_{(s,a)} = \cU^{\cI}_{(s,a)}$.
    Let $\vec{u} \in \cU^{\cLinf}_{(s,a)}$, and as before $\vec{p}_{sa} = P(s,a)$.
    Then $\vec{u}$ by definition satisfies the following constraints
    \begin{align*}
    &\bigwedge_{s' \in S} \vec{u}(s') \in [0,1] \wedge \sum_{s'} \vec{u}(s') = 1 \wedge \norm{\vec{p}_{sa} - \vec{u}}{\infty} \leq \epsilon(s,a).
    \end{align*}
    The first two constraints also appear in the definition of $\cU^{\cI}_{(s,a)}$. 
    We rewrite the last constraint.
    \[
    \norm{\vec{p}_{sa} - \vec{u}}{\infty} \leq \epsilon(s,a) \implies 
    \max_{s'} \big\{|\vec{p}_{sa}(s') - \vec{u}(s')|\big\} \leq \epsilon(s,a).
    \]
    Assume the maximum is attained at some $s^*$, then for all $s' \neq s^*$, we have
    \begin{equation}\label{eq:maximum:linf:element}
    |\vec{p}_{sa}(s') - \vec{u}(s')| \leq |\vec{p}_{sa}(s^*) - \vec{u}(s^*)| \leq \epsilon(s,a).
    \end{equation}
    Additionally, for all $s' \in S$, including $s^*$, we have for arbitrary $c \in \bR_{\geq 0}$ that
    \begin{align*}
     |\vec{p}_{sa}(s') - \vec{u}(s')| \leq c
     &\implies -c \leq \vec{p}_{sa}(s') - \vec{u}(s') \leq c\\ &\implies -c - \vec{p}_{sa}(s') \leq - \vec{u}(s') \leq c - \vec{p}_{sa}(s')\\
     &\implies  c + \vec{p}_{sa}(s') \geq \vec{u}(s') \geq \vec{p}_{sa}(s') - c \\
     &\implies \vec{u}(s') \in [\vec{p}_{sa}(s') -c, \vec{p}_{sa}(s') + c]
    \end{align*}
    From  \Cref{eq:maximum:linf:element} with $c = \epsilon(s,a)$, it then follows that for all $s' \in S$, we have $\vec{u}(s') \in [\vec{p}_{sa}(s') - \epsilon(s,a), \vec{p}_{sa}(s') + \epsilon(s,a)]$. 
    Hence $\forall s \in S, a\in A \colon \cU^{\cLinf}_{(s,a)} \subseteq \cU^{\cI}_{(s,a)}$.

    The other direction follows similarly.
    Assume $\vec{u} \in \cU^{\cI}_{(s,a)}$.
    By definition, $\vec{u}$ satisfies $\forall s' \in S \colon \vec{u}(s') \in [0,1] \wedge \sum_{s'} \vec{u} = 1$, hence we only need to show it satisfies the infinity norm constraint.
    \begin{align*}
        \vec{u} \in \cU^{\cI}_{(s,a)} 
        &\implies \forall s' \in S \colon \vec{u}(s') \in [\vec{p}_{sa}(s') - \epsilon(s,a), \vec{p}_{sa}(s') + \epsilon(s,a)]\\
        &\implies \forall s' \in S \colon |\vec{p}_{sa}(s') - \vec{u}(s')| \leq \epsilon(s,a) \\
        &\implies \max_{s'} \{|\vec{p}_{sa}(s') - \vec{u}(s')|\} \leq \epsilon(s,a).
    \end{align*}
    Thus, $\norm{\vec{p}_{sa} - \vec{u}}{\infty} \leq \epsilon(s,a)$, $\vec{u} \in \cU^{\cLinf}_{(s,a)}$, and $\cU^{\cLinf}_{(s,a)} \supseteq \cU^{\cI}_{(s,a)}$.
    We conclude that both uncertainty sets contain the same probability distributions, and thus that every $\Linf$-RMDP can be reduced to an IMDP.
\end{proof}
From \Cref{lemma:IMDP:is:RMDP} follows that every $\Linf$-RMDP is an $(s,a)$-rectangular RMDP with convex polytopic uncertainty sets.

\subsection{\texorpdfstring{$\Lone$}{L-one}-RMDPs}

\begin{definition}[$\Lone$-RMDP]
Let $\cM = (M,\epsilon)$ where $M = (S,A,P,R,\sinit,\gamma)$ is an MDP, and $\epsilon \colon S \times A \to \bR_{\geq 0}$ is a state-action distance measure.
The $\Lone$-RMDP $\cM$ is an $(s,a)$-rectangular RMDP with state-action uncertainty sets defined by 
\[
\cU_{(s,a)} = \left\{\vec{u} \in \bR_{\geq 0}^{S} \midd \sum_{s' \in S} \vec{u}(s') = 1 \wedge \norm{\vec{p}_{sa} - \vec{u}}{1} \leq \epsilon(s,a)\right\}.
\] 
That is, it is the set of probability vectors within $\Lone$ distance $\epsilon(s,a)$ to the transition probability $P(s,a)$ of $M$.
\end{definition}

The uncertainty sets $\cU_{(s,a)}$ can be described by two types of convex polytopes: one with an exponential number of constraints, and one with $|S|$ additional variables, but a linear number of constraints.
We give both formulations below.

First, the polytope that remains $|S|$-dimensional but incurs an exponential blow-up in the number of constraints.
The polytope is given by
\[
\cU_{(s,a)} = \left\{ \vec{u} \in \bR_{\geq 0}^{S} \midd 
\sum_{s' \in S} \vec{u}(s') = 1 \wedge  \bigwedge_{\vec{t} \in \{-1,1\}^{S}} \vec{t}^\intercal (\vec{u} - \vec{p}_{sa})  \leq \epsilon(s,a) \right\}.
\]
Note the exponential number of constraints, as we add a constraint for each possible vector $\vec{t} \in \{-1,1\}^S$, \ie, $2^{|S|}$ constraints.
It can be represented by the following matrix $\vec{D}_{sa}$ and vector $\vec{b}_{sa}$ such that $\cU_{(s,a)} = (\vec{D}_{sa},\vec{b}_{sa})$:
\[
\vec{D}_{sa} = 
\begin{pmatrix}
    {-}\vec{I}_k \\
    {-}\vec{1}_k^\intercal \\
    \vec{1}_k^{\intercal} \\
    \vec{T}
\end{pmatrix}, \quad \vec{b}_{sa} = \begin{pmatrix}
    \vec{0}_k \\
    {-}1\\
    1\\ 
    \epsilon(s,a)\vec{1}_{2^k} + \vec{T}\vec{p}_{sa}
\end{pmatrix},
\]
where $k = |S|$, $\vec{T} \in \{-1,1\}^{2^{|S|} \times |S|}$ is the matrix of all possible sign vectors, one in each row.

Alternatively, $\cU_{s,a}$ can be lifted to a higher-dimensional polytope $\cU_{s,a}^+$, \ie, one with more variables but a linear number of constraints in the input size.
\[
\cU_{s,a}^{+} = \left\{ (\vec{u},\vec{t}) \in \bR_{\geq 0}^{2|S|} \midd 
\begin{array}{ll}
\sum_{s' \in S} \vec{u}(s') = 1
\wedge \sum_{s' \in S} \vec{t}(s') \leq \epsilon(s,a) \\
 \wedge\bigwedge_{s' \in S} \vec{u}(s') - \vec{t}(s') \leq \vec{p}(s')
 \wedge\bigwedge_{s' \in S} -\vec{u}(s') - \vec{t}(s') \leq -\vec{p}(s') 
\end{array}
\right\}.
\]
The lifted polytope is represented by the following matrix and vector.
\[
\vec{D}_{sa}^+ = 
\begin{pmatrix}
    {-}\vec{I}_k & \vec{0}_k \\
    \vec{0}_k & {-}\vec{I}_k \\
    \vec{1}_k^\intercal & \vec{0}_k \\
    {-}\vec{1}_k^\intercal & \vec{0}_k \\
    \vec{0}_k & \vec{1}_k^\intercal \\
    \vec{I}_k & -\vec{I}_k \\
    {-}\vec{I}_k & {-}\vec{I}_k
\end{pmatrix}, \qquad \vec{b}_{sa}^+ = \begin{pmatrix}
    \vec{0}_k \\
    \vec{0}_k\\
    1\\
    -1\\
    \epsilon(s,a)\\
    \vec{p} \\
    {-}\vec{p}
\end{pmatrix},
\]
where again the subscript $k = |S|$ denotes the respective dimensions. 
We can rewrite the polytope $(\vec{D}_{sa}^+, \vec{b}_{sa}^+)$ as follows.
We split the matrix $\vec{D}_{sa}^+$ into two submatrices $\vec{Y}_{sa}$, $\vec{Z}_{sa}$ such that
\[
\vec{D}_{sa}^+ \begin{pmatrix}
    \vec{u} \\
    \vec{t}
\end{pmatrix} \leq \vec{b}_{sa}^+  \iff 
\vec{Y}_{sa} \vec{u} + \vec{Z}_{sa} \vec{t} \leq \vec{b}_{sa}^+,
\]
where $\vec{Y}_{sa}$ is the first block of columns and $\vec{Z}_{sa}$ the second block of columns of $\vec{D}_{sa}^+$, \ie,
\[
\vec{Y}_{sa} = \begin{pmatrix}
    {-}\vec{I}_k \\
    \vec{0}_k \\
    \vec{1}_k^\intercal \\
    {-}\vec{1}_k^\intercal \\
    \vec{0}_k\\
    \vec{I}_k \\
    {-}\vec{I}_k
\end{pmatrix}, \quad \vec{Z}_{sa} = \begin{pmatrix}
    \vec{0}_k \\
    -\vec{I}_k \\
    \vec{0}_k\\
    \vec{0}_k \\
    \vec{1}_k^\intercal \\
    {-}\vec{I}_k \\
    {-}\vec{I}_k
\end{pmatrix}.
\]
The problem then dualizes into
    \begin{align}
        \min_{\vec{v}} \,&\, {-}\vec{c}^\intercal \vec{v} \notag\\
        \text{Subject to: } \, & \bigwedge_{ s \in S} \begin{cases}
            \vec{e}_s^\intercal \vec{v} + (\vec{b}_s^\pi)^\intercal \vec{w}_s \leq R^\pi(s)  \\
            (\vec{Y}_{s}^\pi)^\intercal \vec{w}_s = \gamma \vec{v} \\
            (\vec{Z}_s^\pi)^\intercal \vec{w}_s = 0\\
            \vec{w}_s \geq 0 
        \end{cases} \label{eq:dtmc:robust:l1:lp:dual}
    \end{align}
This LP follows the same structure as the standard robust LP for a convex polytope detailed in \Cref{eq:dtmc:robust:lp:classical-dual}, except that we add an additional constraint on the second part of the lifted polytope such that the dual projects onto the first part.

\subsection{Kantorovich RMDP}
We close with the RMDPs defined in \Cref{def:bisim:RMDP} that capture bisimulation metrics between two MDP states.
The uncertainty sets are given by the sets of couplings between distributions at each state-state-action tuple:
\[
\cU_{(s,t,a)} = \Lambda_{P(s,a),P(t,a)},
\]
which is a convex polytope.
We now construct the matrix $\vec{D}_{sta}$ and vector $\vec{b}_{sta}$ that represent $\cU_{(s,t,a)}$.
Let $\vec{C}_\mathit{row} = \vec{I}_k \otimes \vec{1}_k^\intercal$ and $\vec{C}_{\mathit{col}} = \vec{1}_k^\intercal \otimes \vec{I}_k$, where $\otimes$ is the Kronecker product between two matrices, $k = |S|$ again, and let $\vec{p}_{sa} = P(s,a)$ and $\vec{p}_{ta} = P(t,a)$ be the probability vectors at state-action pairs $(s,a)$ and $(t,a)$, respectively.
The polytope is then defined by
\[
\vec{D}_{sta} = \begin{pmatrix}
    {-}\vec{I}_{k^2} \\
    {-}\vec{C}_\mathit{row} \\
    \vec{C}_\mathit{row} \\
    {-}\vec{C}_\mathit{col} \\
    \vec{C}_\mathit{col}
\end{pmatrix}, \quad \vec{b}_{sta} = 
\begin{pmatrix}
\vec{0}_{k^2} \\
    {-}\vec{p}_{sa}\\
\vec{p}_{sa}\\
{-}\vec{p}_{ta}\\
\vec{p}_{ta} 
\end{pmatrix}.
\]

\section{Parity-game Hardness}\label{apx:games}
Let $\pi$ be a memoryless policy. We define the \emph{antagonistic value} of the policy by (adversarially) choosing successors in the RMDP. For a function $\tau \colon S \times A \to S$, we say it respects the possibilities in the uncertainty set $\mathcal{U}$, and write $\tau \models \mathcal{U}$ if:
\[
P_{\vec{u}}(s,a,\tau(s,a)) > 0, \text{ for all } s \in S,  a \in \supp(\pi(s)),
\]
for some $\vec{u} \in \mathcal{U}$. 
If, additionally, 
\[
P_{\vec{u}}(s,a,\tau(s,a)) = 1, \text{ for all } s \in S,  a \in \supp(\pi(s)),
\]
for some $\vec{u} \in \cU$, $\tau$ is a vertex of $\cU$ and we write $\tau \in \cU$. 
Now, define:
\[
\mathrm{aVal}^{\pi}_{\cM} = \inf_{\tau \models \mathcal{U}}  \bE_{\pi} \left[
\sum_{t=0}^\infty \gamma^t R(s_t,a_t),
\right],
\]
where 
$s_{t+1} = \tau(s_t,a_t)$ (and thus the only remaining stochasticity, if any, comes from $\pi$).

\begin{lemma}\label{lem:tau-vtx}
    Let $\cU$ be such that, for all $\tau: S \times A \to S$, it contains the $\tau$-vertex. Then, for all policies $\pi$, we have that
    \( \mathrm{aVal}^\pi_{\mathcal{M}} \leq \Vpes_{\mathcal{M}}^\pi.
    \) If, moreover, $\cU$ is the convex hull of all $\tau$-vertices, then we have equality.
\end{lemma}
\begin{proof}
The first part of the statement follows directly from containment of the $\vec{u}$ corresponding to each $\tau$ in $\cU$. For the second part, we note that the $\tau$-vertices correspond to the vertices of the convex hull. Hence, to conclude, recall that the robust discounted reward value of a given policy is achieved at a vertex of the uncertainty polytope.
\end{proof}

A \emph{discounted-sum game}~\cite{DBLP:journals/tcs/ZwickP96} (played on an automaton) is a tuple $(S,A,\Delta,R,s_\iota,\gamma)$ where $S$, $A$, $R$, $s_\iota$ and $\gamma$ are as for MDPs and $\Delta \subseteq S \times A \times S$ is a transition relation. A \emph{strategy} of the protagonist is a function $\sigma \colon S \to A$; while a strategy of the antagonist is a function $\tau: S \times A \to S$ such that $(s,a,\tau(s,a)) \in \Delta$ for all $(s,a) \in S \times A$. We remark that we are restricting ourselves to memoryless deterministic strategies for both players. This is known to be no loss of generality for the value problem (defined below) \cite{DBLP:journals/tcs/ZwickP96}. A pair $(\sigma,\tau)$ of strategies for both players defines a unique infinite path $\rho_{\sigma \tau} = s_0 a_0 s_1 \dots$ from $s_0 = s_\iota$ and its value $\mathrm{Val}(\rho_{\sigma \tau})$ is $\sum^\infty_{t=0}\gamma^t R(s_t,a_t)$. We are usually interested in the value of the game $\mathrm{Val}^* = \max_\sigma \min_\tau \mathrm{Val}(\rho_{\sigma \tau})$. The natural decision problem asks whether this value is at least some given (rational) threshold.
\begin{theorem}
    Discounted-sum games reduce in polynomial time to the robust discounted reward problem, and the constructed RMDP is $(s,a)$-rectangular with convex polytopic uncertainty.
\end{theorem}
\begin{proof}
For any instance of a discounted-sum game, add as uncertainty sets the convex hull of all $\tau$-vertices for $\tau$ a strategy of the adversary. Equivalently, define the uncertainty set so that $\vec{u}(s,a,s') \geq 0$ if $(s,a,s') \in \Delta$, and is $0$ otherwise, for all $s,a,s'$; and $\sum_{s'} \vec{u}(s,a,s') = 1$ for all $s,a$. We observe that this 
yields an $(s,a)$-rectangular uncertainty:  %
the successor probabilities depend  only on $s$ and $a$. Moreover, the uncertainty sets are all convex by definition, too. To conclude, it follows from \autoref{lem:tau-vtx} that the value of the game coincides with the robust expected discounted-sum value of the RMDP. 
\end{proof}
Using the known reduction from parity games to discounted-sum games \cite{DBLP:journals/ipl/Jurdzinski98}, we conclude that if we had a polynomial-time algorithm for the robust expected discounted-sum problem we would have a polynomial-time algorithm for parity games.

\section{Correctness of the Robust Optimality Test}
\label{apx:robust:optimality:test}

For standard MDPs, it is known that the set of possibly optimal actions is defined as
    \[
        A_s^* = \left\{a \in A_s^* \midd \Delta^\pi_{sa} \geq \max_{a \in A} \Delta^\pi_{sa} - \frac{\gamma}{1-\gamma} \Big( \max_{s' \in S} \max_{a' \in A} \Delta^\pi_{s'a'} - \min_{s' \in S} \max_{a' \in A} \Delta^\pi_{s'a'} \Big) \right\}.
    \]
In the following, we show that the set of possibly optimal actions in an RMDP is given by
    \[
        A_s^* = \left\{a \in A_s^* \midd \Deltapes^\pi_{sa} \geq \max_{a \in A} \Deltapes^\pi_{sa} - \frac{\gamma}{1-\gamma} \Big( \max_{s' \in S} \max_{a' \in A} \Deltapes^\pi_{s'a'} - \min_{s' \in S} \max_{a' \in A} \Deltapes^\pi_{s'a'} \Big) \right\}.
    \]
Recall the robust Bellman operator
\[
\fB(\Vpes)(s) = \max_{a \in A} R(s,a) + \gamma \inf_{\vec{u} \in \cU_{(s,a)}} \sum_{s' \in S} P_{\vec{u}}(s,a)(s')\Vpes(s').
\]
The robust Bellman operator is a monotone contraction mapping~\cite[Theorem 3.2]{DBLP:journals/mor/Iyengar05}.

\begin{theorem}
    Given an $(s,a)$-rectangular RMDP $\cM = (S,A,\cU,R,\sinit,\gamma)$, an action $a \in A$ is suboptimal for state $s \in S$, when for any $\vec{x} \in \bR^{S}$
    \[
    R(s,a) + \gamma \inf_{\vec{u} \in \cU_{(s,a)}} \sum_{s' \in S} P_{\vec{u}}(s,a)(s')\vec{x}(s')  < \fB(\vec{x})(s) - \frac{\gamma}{1-\gamma} \mathsf{span}(\fB(\vec{x}) - \vec{x}),
    \]
    where $\mathsf{span}(\vec{y}) = \max_{s \in S} \vec{y}(s) - \min_{s \in S} \vec{y}(s)$ for any $\vec{y} \in \bR^{S}$.
\end{theorem}

\begin{proof}

Let $b,c \in \bR$ with $b \leq c$, $\vec{x} \in \bR^{S}$, and $\vec{v} = \Vpes$ the robust value function as a vector.
If 
\[
\vec{x} + b\vec{1} \leq \vec{v} \leq \vec{x} + c\vec{1} \text{ and }
R(s,a) + \gamma \inf_{\vec{u} \in \cU_{(s,a)}} \sum_{s'} P_{\vec{u}}(s,a)(s') \vec{x}(s') < \fB(\vec{x})(s) - \gamma(c-b),
\]
then action $a$ is suboptimal.

For any $\vec{u} \in \cU_{(s,a)}$, let $\vec{p} \in \bR^{S \times S}$ be the stochastic matrix defined by $\vec{p}_{\vec{u}}^a(s,s') = P_{\vec{u}}(s,a)(s')$.
Then for all $\vec{u} \in \cU_{(s,a)}$
\begin{align*}
    R(s,a) + \gamma \sum_{s' \in S} \vec{p}_{\vec{u}}^a(s,s')(\vec{x}(s') + c) = R(s,a) + \gamma \sum_{s'} \vec{p}_{\vec{u}}^a(s,s')\vec{x}(s') +\gamma c,
\end{align*}
as $\vec{p}_{\vec{u}}^a$ is stochastic.
The above holds in particular for $\inf_{\vec{u} \in \cU_{(s,a)}} \vec{p}_{u}^a$, since $\cU_{(s,a)}$ is a convex polytope and the infimum is attained at one of the vertices. 
Thus, we obtain
\begin{align*}
    R(s,a) + \gamma \inf_{\vec{u} \in \cU_{(s,a)}} \sum_{s' \in S} \vec{p}_{\vec{u}}^a(s,s')(\vec{x}(s') + c) & = R(s,a) + \gamma   \inf_{\vec{u} \in \cU_{(s,a)}} \sum_{s'} \vec{p}_{\vec{u}}^a(s,s')\vec{x}(s') + \gamma c \\
    & \leq \fB(\vec{x})(s) + \gamma b = \fB (\vec{x}+b\vec{1})(s).
\end{align*}

From the monotonicity of $\fB$, we obtain
\begin{align*}
    \vec{v}(s) &= \fB(\vec{v})(s) \geq \fB(\vec{x} + b \vec{1})(s)\\
    &> R(s,a) + \gamma \inf_{\vec{u} \in \cU_{(s,a)}} \sum_{s' \in S} \vec{p}_{\vec{u}}^a(s,s')(\vec{x}(s') + c)\\
    &\geq R(s,a) + \gamma \inf_{\vec{u} \in \cU_{(s,a)}} \sum_{s' \in S} \vec{p}_{\vec{u}}^a(s,s') \vec{v}(s')
\end{align*}

Set $b = \frac{1}{1-\gamma} \min_{s \in S} (\fB(\vec{x}) - \vec{x})(s)$ and $c = \frac{1}{1-\gamma}\max_{s \in S} (\fB(\vec{x}) - \vec{x})(s)$, and we obtain
    \[
    R(s,a) + \gamma \inf_{\vec{u} \in \cU_{(s,a)}} \sum_{s' \in S} P_{\vec{u}}(s,a)(s')\vec{x}(s')  < \fB(\vec{x})(s) - \frac{\gamma}{1-\gamma} \mathsf{span}(\fB(\vec{x}) - \vec{x}),
    \]
    as desired.
\end{proof}

From the result above, with $\vec{x} = \Vpes^\pi$ the value under some policy $\pi$, we note that
\[
(\fB(\Vpes^\pi) - \Vpes^\pi)(s) = \max_{a \in A} R(s,a) + \gamma \inf_{\vec{u} \in \cU_{(s,a)}} \sum_{s' \in S} P_{\vec{u}}(s,a)(s')\Vpes^\pi(s') - \Vpes^\pi(s) = \max_{a \in A} \Deltapes_{sa}^\pi.
\]
Thus, we can reduce the set of possibly optimal actions at state $s \in S$ to
\[
\left\{a \in A \midd \Deltapes_{sa}^\pi \geq \max_{a' \in A} \Deltapes_{sa'}^\pi - \frac{\gamma}{1-\gamma} (\max_{s' \in S} \max_{a' \in A} \Deltapes_{s'a'}^\pi - \min_{s' \in S} \max_{a' \in A} \Deltapes_{s'a'}^\pi)  \right\}.
\]

\section{Proof \texorpdfstring{\Cref{lemma:robust:bisim:value:is:pseudometric}}{Lemma 5.5}}\label{apx:proof:robust:bisim:value:is:pseudometric}
\textbf{\Cref{lemma:robust:bisim:value:is:pseudometric}}\hspace{1em}
The robust value function $\Vpes$ of the RMDP $\cM$ constructed in \Cref{def:bisim:RMDP} is a pseudometric over the set of states $S$.

\begin{proof}
    By construction of the reward function $\cR$, all rewards are nonnegative, hence, $\Vpes \colon S \times S \to \bR_{\geq 0}$, and we can initialize the value function as $\Vpes^{(0)} = \vec{0}$.
    We now show the properties of a pseudometric hold by induction on the iterations of $\Vpes^{(n)}$.
    For the base case $\Vpes^{(0)} = \vec{0}$, the properties immediately follow.
    As induction hypothesis, assume $\Vpes^{(n)}$ is a pseudometric.
    We now show that applying the robust Bellman operator on $\Vpes^{(n)}$, \ie, computing $\Vpes^{(n+1)} = \fB(\Vpes^{(n)})$ again satisfies properties 1--3.
    \begin{enumerate}
        \item Assume $s = s'$, we write
        \[
        \Vpes^{(n+1)}(s,s') = \max_{a} \cR(s,s',a) + \gamma \inf_{\vec{u} \in \cU_{(s,s',a)}} \sum_{(t,t') \in S \times S} P_{\vec{u}}(s,s',a)(t,t')\Vpes^{(n)}(t,t'),
        \]
        and observe that for all $a \in A$ we have $\cR(s,s',a) = (1-\gamma)|R(s,a) - R(s',a)| = 0$.
        Additionally, since $s = s'$, we have $P(s,a) = P(s',a)$ and since $\cU_{(s,s',a)} = \Lambda_{P(s,a),P(s',a)}$ we get that there exists a diagonal coupling of cost zero.
        Hence, $\Vpes^{(n+1)}(s,s') = 0$.
        \item 
        Note that $\Vpes^{(n)}(s,s') = \Vpes^{(n)}(s',s)$, it follows that
        \begin{align*}
            &\Vpes^{(n+1)}(s,s') = \max_{a \in A} \cR(s,s',a) + \gamma \inf_{\vec{u} \in \cU_{(s,s',a)}} \sum_{(t,t') \in S \times S} P_{\vec{u}}(s,s',a)(t,t')\Vpes^{(n)}(t,t') \\
            &\quad = \max_{a \in A} \cR(s',s,a) + \gamma \inf_{\vec{u} \in \cU_{(s',s,a)}} \sum_{(t',t) \in S \times S} P_{\vec{u}}(s',s,a)(t',t)\Vpes^{(n)}(t',t) = \Vpes^{(n+1)}(s',s).
        \end{align*}
        \item Using that $\Vpes^{(n)}$ is a pseudometric by assumption, and that $\cU_{(s,s'',a)}$ is the set of couplings between $P(s,a)$ and $P(s'',a)$, we observe that the inner minimization problem defines the Kantorovich distance, which is a pseudometric and thus satisfies the triangle inequality.
            Thus, it follows that
        \begin{align*}
            &\Vpes^{(n+1)}(s,s'') = \max_{a \in A} \cR(s,s'',a) + \gamma \inf_{\vec{u} \in \cU(s,s'',a)} \sum_{(t,t'') \in S \times S} P_{\vec{u}}(s,s'',a)(t,t'') \Vpes^{(n)}(t,t'')\\
            & \leq \max_{a \in A} \cR(s,s',a) + \cR(s',s'',a) + \gamma \Bigg( \inf_{\vec{u} \in \cU(s,s',a)} \left\{ \sum_{(t,t') \in S \times S} P_{\vec{u}}(s,s',a)(t,t') \Vpes^{(n)}(t,t')\right\} \\
            &\quad + \inf_{\vec{u} \in \cU_{(s',s'',a)}} \bigg\{\sum_{(t',t'') \in S \times S} P_{\vec{u}}(s',s'',a)(t',t'')\Vpes^{(n)}(t',t'') \bigg\} \Bigg)\\
            & \leq \max_{a \in A} \cR(s,s',a) + \gamma \inf_{\vec{u} \in \cU(s,s',a)} \sum_{(t,t') \in S \times S} P_{\vec{u}}(s,s',a)(t,t') \Vpes^{(n)}(t,t')\\
            &\quad + \max_{a \in A} \cR(s',s'',a) + \gamma \inf_{\vec{u} \in \cU(s',s'',a)} \sum_{(t',t'') \in S \times S} P_{\vec{u}}(s',s'',a)(t',t'') \Vpes^{(n)}(t',t'')\\
            & = \Vpes^{(n+1)}(s,s') + \Vpes^{(n+1)}(s',s'').
        \end{align*}
    \end{enumerate}
\end{proof}

\end{document}